\documentclass[11pt]{article}
\usepackage{amsmath}
\usepackage{graphicx,psfrag,pst-node,subfigure,verbatim}
\usepackage{rotating}
\usepackage{mhequ}
\usepackage{amsmath,bbm,amssymb}
\usepackage{amsthm}
\usepackage{float}
\usepackage{graphics}
\usepackage{epsfig}

\newtheorem{thm}{Theorem}[section]
\newtheorem{cor}{Corollary}[thm]

\newtheorem{lemma}{Lemma}[section]

\newtheorem{remark}{Remark}[section]

\def \be{\begin{equs}}
\def \ee{\end{equs}}

\def \E{\mathbb{E}}
\def \ym{{\bm Y}^{\mathrm{mis}}}
\def \yo{{Y}^{\mathrm{obs}}}
\def \yov{{\bm Y}^{\mathrm{obs}}}
\def \yob{\bar{{Y}}^{\mathrm{obs}}}
\def \tfs{\bar{\tau}_{.j}} % Finite population estimand with vector g*
\def \tss{\bar{\tau}_{.j}^{\mathrm{SP}}} % Super population estimand with vector g*
\def \tfss{\bar{\tau}_{.j^{\prime}}} % Finite population estimand with vector g**
\def \hattfs{\widehat{\bar{\tau}}_{.j}} % Finite population estimator with g*
\def \hattfss{\widehat{\bar{\tau}}_{.j^{\prime}}} % Finite population estimator with g**
\def \Zplus{\mathbb{Z}^+} %
\def \Zminus{\mathbb{Z}^-} %

\def \var{\mathrm{Var}}
\def \mm {\mathbf{m}}

\def \cov{\mathrm{Cov}}

\newcommand{\bm}[1]{\mbox{\boldmath$ #1 $\unboldmath}}

\title{Causal inference from $2^k$ factorial designs using the potential outcomes model}

\author{Tirthankar Dasgupta, Natesh S. Pillai \& Donald B. Rubin \\
{\textit{Department of Statistics}}, \\ {\textit{Harvard University, Cambridge, MA 02138}} }

\begin{document}
%\ifpdf
%	\DeclareGraphicsExtensions{.pdf, .jpg, .tif}
%	\else
%	\DeclareGraphicsExtensions{.eps, .jpg}
%	\fi
	
\maketitle

%\keywords{Dictionary learning; Gaussian process; heteroscedastic; multivariate regression; Wishart process.}

\begin{abstract}
A framework for causal inference from two-level factorial designs is proposed. The framework utilizes the concept of potential outcomes that lies at the center stage of causal inference and extends Neyman's repeated sampling approach for estimation of causal effects and randomization tests based on Fisher's sharp null hypothesis to the case of 2-level factorial experiments. The framework allows for statistical inference from a finite population, permits definition and estimation of estimands other than ``average factorial effects'' and leads to more flexible inference procedures than those based on ordinary least squares estimation from a linear model.
\end{abstract}

{\small \textbf{\emph{Keywords: } Factorial Designs, Experimental Design, Potential Outcomes, Causal Inference.
}}

\section{Introduction}
Factorial designs were originally developed in the context of agricultural experiments (Yates 1937, Fisher 1942). These designs allow the relative effects of several factors and their interactions to be assessed efficiently. A $2^K$ factorial design involves $2^K$ treatment combinations arising from $K$ factors, each at two levels. The estimands, or objects of interest, in $2^K$ factorial designs are typically the $2^K-1$ \emph{factorial effects}, which are essentially causal effects of the experimental factors on the response of interest (say $Y$), and include $K$ main effects and ${K \choose j}$ $j$-factor interactions for $j = 2, \ldots, K$.

It is surprising that the estimands have typically received less attention in factorial design literature compared to their estimators. In several textbooks on experimental design (e.g. Wu and Hamada 2009, Ch. 4), estimated factorial effects are first defined as orthogonal contrasts of the observed values of the response. The estimands are then defined as the expectation of these estimators over a hypothetical super population of interest, from which the experimental units are assumed to be randomly sampled. More precisely, the objects of inference are described in terms of the parameters of a regression model of the observed response on the design matrix with an additive error. The inference for factorial effects is, therefore, typically based on a linear model or a generalized linear model (GLM) framework depending on whether the observed response is normally or non-normally distributed.

The linear model based framework for statistical analysis of factorial experiments, although used in most factorial design applications (see, for example, Box Hunter and Hunter 2005, Wu and Hamada 2009), suffers from the following inherent drawbacks:
\begin{enumerate}
\item It defines the causal estimands as a parameter of the probability distribution of the observed response, whereas an estimand should ideally be based on the scientific goals of the experiment.
\item It does not relate the estimands to the population of interest, and more specifically, to the experimental units that constitute the population. For example, it fails to distinguish between situations where the population of interest is a hypothetical super population and infinite (e.g., generated by a manufacturing process) and finite (e.g., a study that pertains to a well-defined set of schools in New York). Whereas the model may make sense under sampling of experimental units from a super population, it is more difficult to interpret the parameters as objects of interest for a finite population.
\item Because the linear model does not include a treatment assignment variable, it needs to be re-defined under every distinct randomization scheme (e.g., typically random effects are introduced for split-plot models (Wu and Hamada 2009)).
\item It does not reveal how the estimation will be affected if additivity does not hold, i.e., if the effects of treatment combinations vary across experimental units.
\end{enumerate}

A natural question that arises at this point is: why have recent developments in the field of factorial experiments ignored the limitations stated above? A plausible explanation emerges from the fact that most of the recent theoretical developments have been triggered by industrial applications, and the linear-model-based framework works well in most typical industrial experiments. First, the experimental units are often exchangeable due to controlled experimental conditions, making the assumption of constant treatment effects a realistic one. Second, the population of interest is typically infinite and hypothetical, because the experimenter's interest lies in the future selection of optimal conditions identified through experimentation. Third, the average treatment effect over the hypothetical population of interest is useful and easily interpretable in most industrial applications.

Unfortunately, due to the above limitations, factorial designs, although being widely used for industrial experimentation, have found few applications in social, behavioral and biomedical experiments, where most of the aforementioned complexities are predominant. Modern-day technology is also making industrial experimentation increasingly complex. Several industrial experiments now have complex randomization structures. Large inherent variation among experimental units, e.g., substrates in nanotechnology, are reportedly leading to poor reproducibility of experimental results (Dasgupta et al. 2008). Considering the potential to apply factorial designs in social, medical and behavioral experiments (Chakraborty et al. 2009, Collins et al. 2009) and their ever-increasing need in engineering and industrial experimentation, a unified framework addressing the limitations of the linear-model based framework for inference appears highly desirable.

In this article, we propose a framework for causal inference from $2^K$ factorial design to address these deficiencies. The framework utilizes the concept of potential outcomes that lies at the center stage of causal inference (Rubin 2008). Although such a framework for single-factor experiments with two levels is well-developed and popularly known as the Rubin Causal Model (RCM, Holland 1986), it is not yet developed for multiple-factor experiments. We consider $2^K$ designs because they can be studied as natural extensions of treatment-control studies with a single factor.

The article is organized as follows. In Section 2 we briefly review the RCM for randomized experiments and observational studies involving a single factor at two levels and discuss it from a historical perspective. In Section 3 we develop the basic framework for extending the RCM to $2^K$ designs by defining factorial effects in terms of the potential outcomes, and propose a response model that is random only through the assignment mechanism. In Sections 4 and 5 respectively, we extend Neymanian (1923) and Fisherian (1925) randomization-based inference to $2^K$ designs. In Section 6 we develop a more flexible Bayesian approach for causal inference from $2^K$ designs. In Section 7, we demonstrate the advantages of the approach with a practical example. Some concluding remarks are given in Section 8.

\section{A brief review of potential outcomes, RCM and its evolution}

The first formal notation for potential outcomes was introduced by Neyman (1923) for randomization-based inference in randomized experiments, and subsequently used by Kempthorne (1952) and Cox (1958) for causal inference from randomized experiments. The concept was formalized and extended by Rubin (1974, 1975, 1977, 1978) for other forms of causal inference from randomized experiments \emph{and} observational studies, and exposition of this transition appears in Rubin (2010).

The early evolution of the RCM was motivated by the need for a clear separation between the object of inference -- often referred to as \emph{the science} -- and what researchers do to learn about the science (e.g., randomly assign treatments to units). In the context of causal inference, the science is a matrix where the rows represent $N$ units, which are physical objects at a particular point of time, and the columns represent potential outcomes under each possible exposure. Thus, for a study with one factor at two levels represented by 1 and 0, each row of the science can be written as $[Y_i(1), Y_i(0)]$, where $Y_i(x)$ is the potential outcome of unit $i$ if unit $i$ receives treatment $x$, $x \in \{0,1\}$, indicated by $W_i(x) = 1$. The representation of the science is adequate under Stable unit treatment value assumption (SUTVA) as defined by Rubin (1980).

 The causal effect of Treatment 1 versus Treatment 0 for the $i$th unit is the comparison of the corresponding potential outcomes for that unit: $Y_i(1)$ versus $Y_i(0)$ (e.g., their difference or their ratio). The ``fundamental problem facing inference for causal effects'' (Rubin 1978) is that only one of the potential outcomes can ever be observed for each unit. Therefore, unit-level causal effects cannot be known and must be inferred. The RCM permits prediction of unit-level causal effects from either the Neymanian (1923/1990) perspective or from the Bayesian perspective (Rubin 1978), although such estimations are generally imprecise relative to the estimation of population or subpopulation causal effects.

Although the average causal effects $\bar{Y}(1) - \bar{Y}(0)$ are common estimands in many fields of application, other summaries of unit-level causal effects may also be of interest in specific scientific studies. As has been emphasized by Rubin (1974-2010), there is no reason to focus solely on average causal effects, although this quantity is especially easy to estimate unbiasedly with standard statistical tools in randomized experiments under simple assumptions.

Rubin (2010) describes the RCM in terms of three legs -- the first being to define causal effects by using potential outcomes (define the science), the second being to describe the process by which some potential outcomes will be revealed (the assignment mechanism), and the third being the Bayesian posterior predictive distribution of the missing potential outcomes.

\vspace{0.5 in}

\section{Extending the RCM to $2^K$ factorial designs} \label{sec:RCM2^k}

\subsection{Defining factorial effects in the potential outcomes framework for a finite population} \label{subsec:factorialeffectsfinite}

Let the $K$ factors be indexed by $k$, and let ${\bm z}$ denote a particular treatment combination represented by a $K$-dimensional vector of -1's and 1's, where the $k$th element indicates whether the $k$th factor is at its low level (-1) or at its high level (+1), $k = 1, \ldots, K$. The number of possible values of ${\bm z}$ is $J = 2^K$; and we denote the set of all $J$ treatment combinations by $\mathbb{Z}$.  Let $Y_i({\bm z})$ denote the potential outcome of the $i$th unit if exposed to treatment ${\bm z}$. Note that when introducing this notation, we implicitly make the stable unit treatment value assumption (SUTVA), which means that the potential outcome of a particular unit depends only on the treatment combination it is assigned, and NOT on the assignments of the remaining units, and that there are no hidden versions of treatments not represented by the $J$ values of ${\bm z}$. All potential outcomes for unit $i$ comprise the vector ${\bm Y}_i$ of dimension $J$. Finally, we define \emph{the science} as the $N \times 2^K$ matrix ${\bm Y}$ of potential outcomes in which the $i$th row is the $J$-vector ${\bm Y}_i$, $i = 1, \ldots, N$, and assume that $N = r 2^K$, where $r$ is an integer representing the number of replications of each treatment combination.

We are interested in contrasting, for each unit, one half of the potential outcomes with the other half of the potential outcomes. For example, the difference of the averages of the potential outcomes when factor 1 is at its high level and at its low level, the so-called ``main effect of factor 1''. Of course, other definitions could be the difference in medians, or the difference in the logarithm of the averages, but the tradition in the study of factorial experiments is to deal with the difference of averages, and this is the situation we focus on here. A factorial effect for each unit is the difference in the averages between one half of the potential outcomes and the other half. Therefore a factorial effect for unit $i$ can be defined by a vector of dimension $J$, say ${\bm g}$, with one half of its elements equal to -1, and one half equal to +1. Dividing the inner product of vector ${\bm g}$ with vector ${\bm Y}_i$ by $2^{K-1}$, we obtain the ${\bm g}$-factorial effect for unit $i$. There are $J-1$ such vectors ${\bm g}$, which we index by $j = 1, \ldots, J-1$ and order as follows. The first, ${\bm g}_1$, generates the main effect of factor 1; ${\bm g}_2$ generates the main effect of factor 2; and proceeding this way, ${\bm g}_K$ generates the main effect of factor $K$. The definition of the main effect of factor $j$ for unit $i$, in terms of potential outcomes is therefore $\tau_{ij} = 2^{-(K-1)}{\bm g}_j^{\prime} {\bm Y}_i, \ i = 1, \ldots, N; \ j = 1, \ldots, K.$

After the main effects, the next ${K \choose 2}$ vectors ${\bm g}_j$ generate the ${K \choose 2}$ two-factor interactions between factors for unit $i$. The first generates the interaction between factor 1 and factor 2, the second generates the interaction between factor 1 and factor 3, etc., in the analogous order, moving on to the two-way interaction between factor $K-1$ and factor $K$. The definition of the two-way interaction between factor $k$ and factor $k^{\prime}$ for unit $i$ is therefore
 $\tau_{ij} = 2^{-(K-1)}{\bm g}_j^{\prime} {\bm Y}_i, \ i=1, \ldots, N; \ j = K+1. \ldots, K+{K \choose 2},$ where each element of ${\bm g}_j$ is obtained by multiplying the corresponding elements of the ${\bm g}$-vectors representing the main effects of factors $k$ and $k^{\prime}$.

The next ${K \choose 3}$ values of ${\bm g}_j$ generate the ${K \choose 3}$ three-way interactions for unit $i$, first between factor 1, factor 2, and factor 3; followed by between factor 1, factor 2, and factor 4; etc. Each ${\bm g}$-vector representing an interaction between three factors is generated by multiplying the three ${\bm g}$-vectors representing the main effects of these factors element-wise. Proceeding in a similar manner, we finally end with the $K$-way interaction between all of the of factors, generated by a vector labeled as ${\bm g}_{J-1}$. For completeness and later use, we label the vector generating the $i$th unit's mean potential outcome as ${\bm g}_0 = (1, \ldots, 1)$, so that we have $J$ values of ${\bm g}_j$, $j = 0, \ldots, J-1$.

A general factorial effect for unit $i$ defined by vector ${\bm g}_j$ will be denoted by $\tau_{ij}$, which is simply a linear combination of unit $i$'s potential outcomes and can be expressed as
\begin{equation}
\tau_{ij} = 2^{-(K-1)}{\bm g}_j^{\prime} {\bm Y}_i, \ i=1, \ldots, N, \;\ j=1, \ldots, {J-1}. \label{eq:defunitfactorial}
\end{equation}

%Note that corresponding to each vector ${\bm g}_*$, their exists a unique disjoint partition $\{ \mathbb{Z^+}({\bm g^*}),\mathbb{Z^-}({\bm g^*})\}$ of the set %$\mathbb{Z}$, so that the factorial effect $\theta_i ({\bm g}_*)$ can be expressed in the following alternative form:
%\begin{equation}
%\theta_i ({\bm g}_*) = 2^{-(k-1)} [ \sum_{z \in \Zplus}Y_i({\bm z}) - \sum_{z \in \Zminus}Y_i({\bm z})]. \label{eq:altdefunitfactorial}
%\end{equation}

The average of these $N$ unit-level factorial effects across the $N$ units for each possible ${\bm g}_j$ define the $J-1$ factorial effect estimands in the finite population:

\begin{equation}
\tfs = \frac{1}{N} \sum_{i=1}^N \tau_{ij} = 2^{-(K-1)}{\bm g}_j^{\prime} \bar{\bm Y}, \label{eq:altdefpopfactorial}
\end{equation}

where
\begin{equation}
\bar{\bm Y} =  \frac{1}{N} \sum_{i=1}^N {\bm Y}_i. \label{eq:ybar(x)}
\end{equation}

The factorial effects $\tfs$ defined by (\ref{eq:altdefpopfactorial}) are the \emph{finite-population estimands} in the statistical inference problem. Note that apart from defining the average factorial effect $\tfs$, the potential outcomes model and the definition of unit-level factorial effects permit defining the median of unit-level factorial effects or the quantiles of the distribution of the unit-level factorial effects as population-level factorial effects. Such population-level effects may be more meaningful estimands in certain scientific queries compared to the conventional average factorial effects.

%Note that the factorial effect $\tfs$ can also be expressed in the following form which is analogous to (\ref{eq:altdefunitfactorial}):
%\begin{equation}
%\theta_i ({\bm g}_*) = 2^{-(k-1)} [ \sum_{z \in \Zplus}\bar{Y}({\bm z}) - \sum_{z \in \Zminus}\bar{Y}({\bm z})], \label{eq:defpopfactorial}
%\end{equation}
%where $\bar{Y}({\bm z}) = N^{-1} \sum_{i=1}^N Y_i({\bm z})$.

\textbf{Example}: Consider a $2^2$ factorial experiment, where the potential outcomes ${\bm Y}_i$ for the $i$th unit are  $Y_i(-1,-1), Y_i(-1,1), Y_i(1,-1), Y_i(1,1)$. The main effects of factors 1 and 2, respectively, are represented by vectors ${\bm g}_1 = (-1,-1,1,1)^{\prime}$, ${\bm g}_2 = (-1,1,-1,1)^{\prime}$, and the two-way interaction between factors 1 and 2 is represented by the vector ${\bm g}_3 =(1,-1,1,-1)^{\prime}$. The three unit-level factorial effects $\tau_{i1}, \tau_{i2}$ and $\tau_{i3}$ are given by $2^{-1} {\bm g}_1^{\prime} {\bm Y}_i$, $2^{-1}{\bm g}_2^{\prime} {\bm Y}_i$ and $2^{-1} {\bm g}_3^{\prime} {\bm Y}_i$, and the corresponding population-level factorial effects $\bar{\tau}_{.1}, \bar{\tau}_{.2}$ and $\bar{\tau}_{.3}$ are $2^{-1} {\bm g}_1^{\prime} \bar{\bm Y}$, $2^{-1}{\bm g}_2^{\prime} \bar{\bm Y}$ and $2^{-1} {\bm g}_3^{\prime} \bar{\bm Y}$, where $\bar{\bm Y}$ is given by (\ref{eq:ybar(x)}).

The definitions of unit-level and finite population-level factorial effects can easily be extended to define factorial effects for a hypothetical \emph{super} population (Imbens and Rubin 2011, Ch. 6) by letting the population size $N$ grow infinitely large, keeping the number of factors $K$ fixed. Super population-level factorial effects are defined as
\begin{equation}
\tss = \lim_{r \rightarrow \infty} \tfs  = \lim_{N \rightarrow \infty} \frac{1}{N} \sum_{i=1}^N \tau_{ij}, \label{eq:limitfacteffect}
\end{equation} \label{def:superpop1}
where $\tau_{ij}$ and $\tfs$ are defined by (\ref{eq:defunitfactorial}) and (\ref{eq:altdefpopfactorial}) respectively.

\subsection{Assignment mechanism, observed outcomes and unbiased estimation of treatment means}

As observed repeatedly by Rubin (2005, 2008), statistical inference for causal effects requires the specification of an assignment mechanism; a probabilistic model for how experimental units are allocated to the treatment combinations. For $2^K$ factorial designs, we define the following assignment variables:
  \[ W_i({\bm z}) = \left \{ \begin{array}{cl}
  1 & \mbox{if the} \;\ i\mbox{th} \;\ \mbox{unit is assigned to} \;\ {\bm z} \\
  0 &  \mbox{otherwise}.
  \end{array}  \right. \label{eq:assignmentvar} \]
\noindent Clearly, there are $2^k N$ such assignment variables satisfying the following conditions:
\begin{eqnarray}
\sum_{{\bm z}} W_i({\bm z}) &=& 1, \;\ i=1, \ldots, N, \label{eq:wrest1} \\
\sum_{i=1}^N W_i({\bm z}) &=& r, \;\ \mbox{for all} \;\ {\bm z}. \label{eq:wrest2}
\end{eqnarray}
Condition (\ref{eq:wrest1}) follows from the fact that every experimental unit can be allocated to one and only one treatment combination, whereas (\ref{eq:wrest2}) specifies the number of replications ($r$) for each treatment combination. Depending on the treatment assignment mechanism, $W_i({\bm z})$ can have different probability distributions. We assume a completely randomized treatment assignment so that $\E(W_i({\bm x})|{\bm Y}) = 1/2^K$ for all ${\bm z}$, where $\E(\cdot|{\bm Y})$ is the expectation over all possible randomizations.

Denote the observed outcome corresponding to the $i$th experimental unit by $\yo_i , i=1, \ldots, N$, so that
\begin{equation}
\yo_i = \sum_{\bm z} W_i({\bm z})Y_i({\bm z}). \label{eq: observed outcome}
\end{equation}
\noindent
A natural estimator of $\bar{Y}({\bm z})$, the population mean for treatment combination ${\bm z}$, is
\begin{equation}
\yob({\bm z}) = \frac{1}{r} \sum_{i: W_i({\bm z})=1} \yo_i = \frac{1}{r} \sum_{i=1}^N W_i({\bm z}) Y_i({\bm z}), \label{eq: ybar}
\end{equation}
\noindent which is an unbiased estimator for $\bar{Y}({\bm z})$ because $\E(W_i({\bm z})|{\bm Y}) = 1/2^k = r/N$. The following three lemmas are used to derive important sampling properties of $\yob({\bm z})$, and consequently, those of the estimated factorial factorial effects, defined in Section \ref{sec:neyman}. The first two lemmas concern the assignment variable $W_i({\bm z})$. Define
\begin{equation}
D_i({\bm z}) = W_i({\bm z}) - \frac{r}{N}, \label{def D}
\end{equation}
where $\E(D_i({\bm z})|{\bm Y}) = 0$. Then,
\begin{lemma}
For $i, i^{\prime} = 1, \ldots, N$ and any treatment combination ${\bm z}$,
 \[ \E \big( D_i({\bm z}),D_{i^{\prime}}({\bm z}) | {\bm Y} \big ) = \left \{ \begin{array}{cc}
  \frac{r(N-r)}{N^2} & \mbox{if} \;\ i=i^{\prime}  \\ \\
  \frac{-r(N-r)}{N^2(N-1)} & \mbox{if} \;\ i \ne i^{\prime}
  \end{array}  \right. \] \label{lem: lemma 1}
\end{lemma}

\begin{lemma}
For $i ,i^{\prime} = 1, \ldots, N$ and any two distinct treatment combinations ${\bm z}$ and ${\bm z}^*$, i.e., ${\bm z} \ne {\bm z}^*$
 \[ \E \big( D_i({\bm z}),D_{i^{\prime}}({\bm z}^*) | {\bm Y} \big ) = \left \{ \begin{array}{cc}
  \frac{-r^2}{N^2} & \mbox{if} \;\ i=i^{\prime}  \\ \\
  \frac{r^2}{N^2(N-1)} & \mbox{if} \;\ i \ne i^{\prime}
  \end{array}  \right. \] \label{lem: lemma 2}
\end{lemma}

\begin{lemma}
Let
\begin{equation}
S^2({\bm z}) = \sum_{i=1}^N (Y_i({\bm z}) - \bar{Y}({\bm z}))^2 / (N-1), \label{eq:varx}
\end{equation}
denote the variance of all potential outcomes (with divisor $N-1$) for treatment combination ${\bm z}$, and
\begin{equation}
S^2({\bm z}, {\bm z}^*) = \sum_{i=1}^N (Y_i({\bm z}) - \bar{Y}({\bm z}))(Y_i({\bm z}^*) - \bar{Y}({\bm z}^*)) / (N-1), \label{eq:covx}
\end{equation}
denote the covariance (with divisor $N-1$) among pairs of potential outcomes for two distinct treatment combinations ${\bm z}$ and ${\bm z}^*$. Then, $S^2({\bm z})$ and
$S^2({\bm z}, {\bm z}^*)$ can be expressed in the following alternative forms:
\begin{eqnarray*}
S^2({\bm z}) &=& \frac{1}{N} \Big[\sum_{i=1}^N Y_i^2({\bm z}) - \frac{1}{N-1} \sum_{i=1}^N \sum_{i^{\prime}=1, i^{\prime} \ne i}^N Y_i ({\bm z}) Y_i^{\prime}({\bm z})  \Big], \label{eq: varxalt} \\
S^2({\bm z},{\mathbf z^*}) &=& \frac{1}{N} \Big[ \sum_{i=1}^N  Y_i({\bm z}) Y_i({\bm z}^*) - \frac{1}{N-1} \sum_{i=1}^N \sum_{i^{\prime}=1, i^{\prime} \ne i}^N Y_i({\bm z}) Y_i^{\prime}({\bm z}^*) \Big] . \label{eq:covxalt}
\end{eqnarray*} \label{lemma:varcovxalt}
\end{lemma}

By application of Lemmas \ref{lem: lemma 1} - \ref{lemma:varcovxalt}, it follows that
\begin{eqnarray}
\var(\yob({\bm z})|{\bm Y}) &=& \frac{N-r}{rN} S^2({\bm z}), \;\ \mbox{for all} \;\ {\bm z} , \label{eq:varybar}\\
\cov(\yob({\bm z}), \yob({\bm z})^*|{\bm Y}) &=& -\frac{1}{N} S^2({\bm z}, {\bm z}^*) \;\ \mbox{for all} \;\ {\bm z}\ne{\bm z}^*, \label{eq:covybar}
\end{eqnarray}
where $S^2({\bm z})$ and $S^2({\bm z}, {\bm z}^*)$ are defined in Lemma \ref{lemma:varcovxalt}.

We conclude this section by observing that an unbiased estimator of any population-level average factorial effect $\tfs$  given by (\ref{eq:altdefpopfactorial}) can be obtained by replacing the treatment mean $\bar{Y}({\bm z})$ by $\yob({\bm z})$. We will now formally define this estimator and study its repeated sampling properties.

\section{Neymanian randomization-based inference for $2^K$ factorial designs}
\label{sec:neyman}

As observed by Rubin (2008), ``Neyman's form of randomization-based inference can be viewed as drawing inferences by evaluating the expectations of statistics over the distribution induced by the assignment mechanism in order to calculate a confidence interval for the typical causal effect.'' Here we study the sampling properties of the estimated factorial effects $\hattfs$ defined as
\begin{equation}
 \hattfs = 2^{-(K-1)} {\bm g}_j^{\prime} \bar{\bm Y}^{\mathrm{obs}}, \label{eq:defpopfactorialest}
\end{equation}
where $\bar{\bm Y}^{\mathrm{obs}}$ is the $J$-vector of average observed outcomes, and investigate the effect of additivity on the sampling properties under the Neymanian framework. Now, we state two theorems (proofs in the Appendix) that summarize the sampling properties of estimated average factorial effects.

\begin{thm}
For a completely randomized treatment assignment mechanism, the estimator $\hattfs$ given by (\ref{eq:defpopfactorialest}) is an unbiased estimator of the factorial effect $\tfs$, that is, averaging over the randomization distribution of the treatment assignment variables. \label{thm: mean}
\end{thm}

\begin{thm}
For a completely randomized treatment assignment mechanism, the sampling variance of $\hattfs$ is
\begin{eqnarray}
\frac{1}{2^{2(K-1)}r} \sum_{{\bm z} } S^2({\bm z}) - \frac{1}{N} S_j^2, \label{eq:varthetahat}
\end{eqnarray}
where $S^2({\bm z})$ is given by (\ref{eq:varx}), and
\begin{equation}
S_j^2 = \frac{1}{N-1} \sum_{i=1}^N ( \tau_{ij} - \tfs )^2, \label{eq:vartheta}
\end{equation}
denotes the population variance of all unit level factorial effects $\tau_{ij}$, $i=1, \ldots, N$.
\label{thm:variance}
\end{thm}

Although Theorem \ref{thm: mean} is intuitive and straightforward, Theorem \ref{thm:variance} extends Neyman's (1923) results to $2^K$ factorial designs and provides us with a way to assess the effect of additivity on the sampling variance of an estimated factorial effect. Clearly, the term $S_j^2$ in the right hand side of (\ref{eq:varthetahat}) cannot be estimated from the experimental data because none of unit-level factorial effects $\tau_{ij}$ are observable. Let
\begin{equation}
\rho({\bm z},{\bm z}^*) = \frac{S^2({\bm z}, {\bm z^*})}{S({\bm z}) S({\bm z^*})}, \label{eq:rhodef}
\end{equation}
be the finite population correlation coefficient among all potential outcomes for treatment combinations ${\bm z}$ and ${\bm z}^*$, where $S^2({\bm z})$ and $S^2({\bm z}, {\bm z^*})$ are given by (\ref{eq:varx}) and (\ref{eq:covx}) respectively. For the $j$th factorial effect, define the subsets $\Zplus_j, \Zminus_j \subset \mathbb{Z}$ by
\be \label{eqn:Z}
\Zplus_j = \{z \in \mathbb{Z}: j\mbox{th element of}\, z  = +1 \}, \quad
\Zminus_j = \{z \in \mathbb{Z}: j\mbox{th element of}\, z  = -1 \} \;.
\ee
Thus  $\Zplus_j$ and $\Zminus_j$ partition the set ${\mathbb Z}$ into two disjoint sets of treatment combinations. It can be shown after some straightforward algebraic manipulations that
\begin{eqnarray}
S_j^2 &=& \frac{1}{2^{2(K-1)}} \Big[ \sum_{z \in \mathbb{Z}} S^2({\bm z}) + 2 \big[ \sum_{{\bm z} \in \Zplus_j} \sum_{{\bm z}^* \in \Zplus_j} \rho({\bm z},{\bm z}^*)S({\bm z})S({\bm z}^*) \nonumber \\
&+& \sum_{{\bm z} \in \Zminus_j} \sum_{{\bm z}^* \in \Zminus_j} \rho({\bm z},{\bm z}^*)S({\bm z})S({\bm z}^*) + \sum_{{\bm z} \in \Zplus_j} \sum_{{\bm z}^* \in \Zminus_j} \rho({\bm z},{\bm z}^*)S({\bm z})S({\bm z}^*) \big] \Big]. \label{eq:s^2theta}
\end{eqnarray}

Recall that two distinct treatment combinations ${\bm z}$ and ${\bm z}^*$ are said to have additive effects if the difference of the unit-level potential outcomes $Y_i({\bm z}) - Y_i({\bm z}^*)$ is the same for each $i = 1, \ldots, N$. Then a set of necessary and sufficient conditions for ``strict'' additivity (i.e., additivity of effects for all treatment combinations) are: (i) $S^2({\bm z}) = S^2$ and (ii) $\rho({\bm z},{\bm z}^*) = 1$ for all ${\bm z},{\bm z}^*$. It immediately follows from (\ref{eq:s^2theta}) that a necessary and sufficient condition for strict additivity is: $ S_j^2 = 0 $. Thus, substituting $S^2({\bm z}) = S^2$ and $ S_j^2 = 0 $ into (\ref{eq:varthetahat}), we arrive at the following corollary

\begin{cor}
Under strict additivity of treatment effects, for any factorial effect $\tfs$, the sampling variance of $\hattfs$ is
\begin{equation}
\var(\hattfs|{\bm Y}) = \frac{4}{N} S^2, \label{eq:varthetahat_additivity}
\end{equation}
where $S^2$ is the common population variance of potential outcomes for each treatment combination ${\bm z}$. \label{cor:varthetahat_additivity}
\end{cor}

It is well known (e.g., Wu and Hamada 2009 Chapter 4) that under the assumption of a linear model with additive independent and identically distributed errors with common variance $\sigma^2$, the sampling variance of estimators of super population-level estimands $\tss$ is $(4/n) \sigma^2$, where $n$ is a finite sample of units chosen randomly for the experiment from the super population. Here we show that an analogous result holds for estimators of finite population-level estimands under the strong assumption of strict additivity.

Next, to explore the effect of non-additivity on sampling variances, we consider a special case where $S^2({\bm z}) = S^2$ and $\rho({\bm z},{\bm z}^*) = \rho (< 1) $ for all ${\bm z}, {\bm z}^*$, which is known as \emph{Compound Symmetry} of the variance-covariance matrix. Substituting the two conditions into (\ref{eq:s^2theta}), we have
\begin{equation}
S_j^2 = \frac{1}{2^{K-2}}(1-\rho) S^2, \label{eq:s^2thetanon1}
\end{equation}
and consequently obtain the following corollary:
\begin{cor} \label{cor:nonadd}
If non-additivity of treatment effects is characterized by compound symmetry of the potential outcomes, then the sampling variance of $\hattfs$ is
\begin{equation}
\var(\hattfs|{\bm Y}) = \frac{4}{N} \big( 1 - \frac{1-\rho}{2^K} \big) S^2 . \label{eq:varthetahat_non1}
\end{equation} \label{cor:varthetahat_non1}
\end{cor}

 Under the assumption of compound symmetry of the potential outcomes, to ensure positive definiteness of the $2^K \times 2^K$ matrix of population variances, $S^2({\bm z})$ and population covariances $S^2({\bm z}, {\bm z}^*)$, we must have $\rho > -1/(2^K-1)$. Therefore, from (\ref{eq:varthetahat_additivity}) and (\ref{eq:varthetahat_non1}), a third corollary on the bounds of $\var(\hattfs)$ can immediately be established:

\begin{cor}
Under compound symmetry, the sampling variance of each estimated factorial effect satisfies the following bounds:
\begin{equation}
\frac{4}{N} (1-\frac{1}{2^K-1}) S^2 < \var(\hattfs|{\bm Y}) \le \frac{4}{N} S^2 . \label{eq:varthetabounds}
\end{equation} \label{cor:varthetabounds}
\end{cor}

Two interesting observations follow from Corollaries \ref{cor:varthetahat_non1} and \ref{cor:varthetabounds}. First, $\var(\hattfs)$ is largest under strict additivity, as observed by Imbens and Rubin (2012, Ch. 6) for a single-factor experiment. Second, under compound symmetry, the effect of non-additivity gets smaller as the number of factors $K$ increases.

As an illustration of how the above results can help understand the effect of lack of additivity on the sampling variance of estimated factorial effects, consider the case of a $2^2$ design with the four treatment combinations (-1,-1), (-1,1), (1,-1) and (1,1). Then under strict additivity, the estimators of the causal effects $\bar{\tau}_{.1}$,$\bar{\tau}_{.2}$ and $\bar{\tau}_{.3}$ have the same sampling variance $S^2/r$, where $r = N/4$ denotes the number of replications. Under the assumption of compound symmetry, the estimators have the same variance: $ (3 + \rho) S^2 / (4r).$

The assumption of a common correlation among potential outcomes of all pairs of treatment combinations is strong. A more realistic potential outcomes model may be one that assumes different correlations associated with different factors. Assume that, in a $2^2$ experiment, the potential outcomes are correlated only if they correspond to a fixed level of factor $1$, and are uncorrelated otherwise. This means, out of the six possible pairs of potential outcomes, only two -- $(Y(-1,-1), Y(-1,1))$ and $(Y(1,-1), Y(1,1))$ -- have a non-zero correlation coefficient $\rho$, and the matrix ${\bm R}$ of all correlation coefficients $\rho({\bm z}, {\bm z}^*)$ has the following structure:
\begin{equation} \label{eqn:Rdef}
{\bm R} = \left(
            \begin{array}{cccc}
              1 & \rho & 0 & 0 \\
              \rho & 1 & 0 & 0 \\
              0 & 0 & 1 & \rho \\
              0 & 0 & \rho & 1 \\
            \end{array}
          \right).
\end{equation}

\noindent Then, by successive application of (\ref{eq:s^2theta}) and Theorem \ref{thm:variance}, it follows that the sampling variance of $\widehat{\bar{\tau}}_{.1}$ equals $(3+\rho)S^2/(4r)$, whereas the sampling variances of $\widehat{\bar{\tau}}_{.2}$ and $\widehat{\bar{\tau}}_{.3}$ are both equal to $(3-\rho)S^2/(4r)$. This example makes it clear that the sampling variances of estimated factorial effects can be severely affected by lack of additivity, and may be different depending on the structure of the correlation matrix ${\bm R}$. Unfortunately, without additional information or assumptions, there is no way one can estimate the correlation matrix from data.

So far we have considered the sampling distribution of estimated factorial effects individually. For factorial designs, it is also important to study the joint distribution of all the estimated factorial effects. The following result gives an expression for the sampling covariance between two estimated factorial effects for a finite population.

\begin{thm}
Let $\tfs$ and $\tfss$ denote two distinct factorial effects. For a completely randomized treatment assignment mechanism, the covariance between the estimated factorial effects $\hattfs$ and $\hattfss$ is
\begin{eqnarray*}
&& \frac{1}{2^{2(K-1)}r} \Big [ \sum_{{\bm z} \in \Zminus_j \bigcap \Zminus_{j^{\prime}} } S^2({\bm z}) - \sum_{{\bm z} \in \Zminus_j \bigcap \Zplus_{j^{\prime}} } S^2({\bm z}) - \sum_{{\bm z} \in \Zplus_j \bigcap \Zminus_{j^{\prime}} } S^2({\bm z}) + \sum_{{\bm z} \in \Zplus_j \bigcap \Zplus_{j^{\prime}} } S^2({\bm z}) \Big] - \frac{1}{N}S_{j j^{\prime}}^2,
\end{eqnarray*}
where $S^2({\bm z})$ is given by (\ref{eq:varx}), and
\begin{equation}
S_{j j^{\prime}}^2 = \frac{1}{N-1} \sum_{i=1}^N [ \tau_{ij} - \tfs] [ \tau_{i j^{\prime}} - \tfss] \label{covartheta}
\end{equation}
denotes the finite population covariance between all unit level factorial effects $\tau_{ij}$ and $\tau_{i j^{\prime}}$, $i=1, \ldots, N$. \label{thm:covariance}
\end{thm}

Note that $\cov(\hattfs,\hattfss)$ has a form that is very similar to the expression for $\var(\hattfs)$ obtained in Theorem \ref{thm:variance} in the sense that it involves the term $S_{j j^{\prime}}^2$ which is non-estimable like $S_j^2$ in Theorem \ref{thm:variance}. Noting that under strict additivity, $S^2({\bm z}) = S^2$ for all ${\bm z} \in {\mathbb Z}$ and $S^2_{j j^{\prime}} = 0$ for all $j \ne j^{\prime}$, we have the following corollary:

\begin{cor}
For all $j \ne j^{\prime}$, estimated factorial effects $\hattfs$ and $\hattfss$ are uncorrelated under strict additivity.
\end{cor}

To explore the effect of non-additivity on $\cov(\hattfs,\hattfss)$, we again consider the case of compound symmetry. The following corollary can be established after some algebraic manipulations (proof in Appendix).

\begin{cor}
For all $j \ne j^{\prime}$, estimated factorial effects $\hattfs$ and $\hattfss$ are uncorrelated under compound symmetry.
\end{cor}

\vspace{0.1 in}

\noindent \textbf{Estimation of variance and covariance and inference on factorial effects}

The Neymanian inference procedure depends on the estimation of the factorial effects and their sampling variances and covariances. To estimate the sampling variances and covariances, we first seek an unbiased estimator of $S^2({\bm z})$. Clearly, the sample variance for the treatment combination ${\bm z}$, given by
\begin{equation}
s^2({\bm z}) = \frac{1}{r-1} \sum_i (\yo_i - \yob({\bm z}))^2, \label{eq:sampvariance}
\end{equation}
is an unbiased estimator of $S^2({\bm z})$. Inserting $s^2({\bm z})$ in place of $S^2({\bm z})$, the Neymanian estimator of $\var(\hattfs)$ can thus be obtained from Theorem \ref{thm:variance} as:
\begin{equation}
\widehat{\var}_{\mathrm{Ney}}(\hattfs | {\bm Y}) = \frac{1}{2^{2(K-1)}r} \sum_{{\bm z} } s^2({\bm z}). \label{eq:sethetahat}
\end{equation}
%\textcolor{red}{It is so not frequentist to condition on $Y$? change the notation above?}
Because $S_j^2$ is not observable and estimable without assumptions, $\widehat{\var}_{\mathrm{Ney}}(\hattfs)$ in expectation is an upper bound on the true sampling variance of $\hattfs$. The bias $S_j^2/N$ depends on the correlation matrix ${\bm R}$ of potential outcomes defined in \eqref{eqn:Rdef}. However, from Theorem 3, it follows that the estimator
\small
\begin{eqnarray}
\widehat{\cov}_{\mathrm{Ney}}(\hattfs,\hattfss | {\bm Y}) = - \frac{1}{2^{2(K-1)}r} \Big [ \sum_{{\bm z} \in \Zminus_j \bigcap \Zminus_{j^{\prime}} } s^2({\bm z}) - \sum_{{\bm z} \in \Zminus_j \bigcap \Zplus_{j^{\prime}} } s^2({\bm z})- \sum_{{\bm z} \in \Zplus_j \bigcap \Zminus_{j^{\prime}} } s^2({\bm z}) + \sum_{{\bm z} \in \Zplus_j \bigcap \Zplus_{j^{\prime}} } s^2({\bm z}) \Big]. \label{eq:covthetahat}
\end{eqnarray}
\normalsize
can overestimate or underestimate $\cov(\hattfs, \hattfss|{\bm Y})$ depending on whether $S^2_{j j^{\prime}}$ is positive or negative.
\vspace{0.1 in}

Let ${\bm \tau}$ denote the vector of the $J-1$ population-level factorial effects. To draw inference about ${\bm \tau}$, consider the statistic:
\begin{equation}
T^N = \hat{\bm \tau}^{\prime} \hat{\Sigma}_{\hat{\bm \tau}}^{-1} \hat{\bm \tau}, \label{statistic_Neyman}
\end{equation}
where $\hat{\bm \tau}$ is the Neymanian estimator of ${\bm \tau}$; $\Sigma_{\hat{\bm \tau}}$ denotes the true covariance matrix of $\hat{\bm \tau}$, and $\hat{\Sigma}_{\hat{\bm \tau}}$ its estimator. The diagonal elements of $\hat{\Sigma}_{\hat{\bm \tau}}$ can be obtained from (\ref{eq:sethetahat}), and the off-diagonals from (\ref{eq:covthetahat}). Note that $\hat{\Sigma}_{\hat{\bm \tau}}$ will be a diagonal matrix under the assumption of equal variance of potential outcomes for all treatment combinations. For $0 < \alpha < 1$, a $100(1-\alpha)\%$ (Neymanian) confidence region for ${\bm \tau}$ is:
\begin{equation}
\big \{ \hat{\bm \tau}: p_{1-\alpha/2} \le T^N \le p_{\alpha/2} \big\}, \label{eq:ci_group}
\end{equation}
where $p_{\alpha}$ denotes the upper-$\alpha$ point of the asymptotic distribution of $T^N$.

From Corollary \ref{cor:varthetahat_additivity}, it follows that under the assumption of strict additivity, $T^N = (N/4) \hat{\bm \tau}^{\prime} \hat{\bm \tau} / \hat{S}^2$, where $\hat{S}^2$ is the estimator of the common variance $S^2$. Substituting $\hat{S}^2 = \sum_{{\bm z}} s^2({\bm z}) / 2^K$, it is easy to see that, under strict additivity, $T^N$ is the ratio of the treatment mean sum of squares (MSTr) and residual mean sum of squares (MSR), as defined in the standard analysis of variance framework (e.g., Montgomery 2000). Following the work of Wilk (1955), it can be argued that if strict additivity holds, the randomization distribution of $T^N$ can be approximated by an $F$ distribution with $J-1, N-J$ degrees of freedom. Under sampling from a normal super population when the inference is about the $J-1$-vector of super-population estimands ${\bm \tau}^{\mathrm{SP}}$ instead of ${\bm \tau}$, this result is exact and much more straightforward to prove.

Also, under a normal approximation (known to be valid under strict additivity), $100(1-\alpha)\%$ (Neymanian) confidence intervals for $\tfs$ can be obtained as
\begin{equation}
\hattfs \pm z_{\alpha/2} \widehat{\mathrm{S.E.}}(\hattfs), \label{eq:ci_indiv}
\end{equation}
 where $\widehat{\mathrm{S.E.}}(\hattfs) = \sqrt{\widehat{\var}_{\mathrm{Ney}} (\hattfs)}$ denotes the estimated standard error of $\hattfs$, which can be obtained from (\ref{eq:sethetahat}), and $z_{\alpha}$ denotes the upper $\alpha$ point of the standard normal distribution.

Neymanian inference on individual or a collection of factorial effects under the proposed framework is essentially the same as the linear-model-based inference that has been widely used by experimenters and analysts for decades (see, for example, Wu and Hamada 2009; Box, Hunter and Hunter 2005). However, as noted by Imbens and Rubin (2012, Ch. 6) for $K=1$, it is imperative that under lack of additivity, the asymptotic confidence intervals given by (\ref{eq:ci_group}) and (\ref{eq:ci_indiv}) are viewed as \emph{conservative} in the sense that they are wider than the ``true'' intervals that involve the actual standard errors, which are overestimated in the Neymanian framework. Knowledge of the correlation matrix ${\bm R}$, defined in \eqref{eqn:Rdef}, can shorten the confidence intervals by an amount that depends on the structure of ${\bm R}$ and its entries. To illustrate this aspect more formally using the example of the $2^2$ experiment discussed earlier, consider the following four correlation structures:

\scriptsize
\[ {\bm R}_1 = \left(
            \begin{array}{cccc}
              1 & 1 & 1 & 1 \\
              1 & 1 & 1 & 1 \\
              1 & 1 & 1 & 1 \\
              1 & 1 & 1 & 1 \\
            \end{array}
          \right),
   {\bm R}_2 = \left(
            \begin{array}{cccc}
              1 & \rho & \rho & \rho \\
              \rho & 1 & \rho & \rho \\
              \rho & \rho & 1 & \rho \\
              \rho & \rho & \rho & 1 \\
            \end{array}
          \right),
   {\bm R}_3 = \left(
            \begin{array}{cccc}
              1 & \rho & 0 & 0 \\
              \rho & 1 & 0 & 0 \\
              0 & 0 & 1 & \rho \\
              0 & 0 & \rho & 1 \\
            \end{array}
          \right),
   {\bm R}_4 = \left(
            \begin{array}{cccc}
              1 & \rho_1 & \rho_2 & 0 \\
              \rho_1 & 1 & 0 & \rho_2 \\
              \rho_2 & 0 & 1 & \rho_1 \\
              0 & \rho_2 & \rho_1 & 1 \\
            \end{array}
          \right).\]
\normalsize
Note that in order to be positive definite, ${\bm R}_1$ must satisfy $-1/3 \le \rho \le 1$ and ${\bm R}_3$ must satisfy $|\rho_1 + \rho_2| < 1$. For compound symmetry, ${\bm R}_2$ is positive definite for all $-1 \le \rho \le 1$. Assume that in all the cases the potential outcomes have the same variance $S^2$ for all four treatment combinations. Table \ref{tab:add_neyman} shows the 95\% confidence intervals (without any corrections for multiple testing) for the three factorial effects $\bar{\tau}_{.1}, \bar{\tau}_{.2}$ and $\bar{\tau}_{.3}$. Lack of knowledge of ${\bm R}$ leads to the same confidence interval $\pm 1.96 \times 2 \frac{s}{\sqrt{N}}$ for each factorial effect on every occasion, whereas better inference (tighter confidence intervals) is possible if such knowledge can be incorporated by conditioning on covariates and replacing partial correlations for correlations.

\begin{table}
\caption{Effect of lack of additivity of inference of factorial effects for a $2^2$ design} \label{tab:add_neyman}
\small
\centering
\begin{tabular}{c|l|c}
\hline
\multicolumn{1}{c|}{Correlation} & \multicolumn{2}{c}{95\% CI for factorial effects under} \\
\multicolumn{1}{c|}{structure}& \multicolumn{1}{c|}{perfect knowledge of ${\bm R}$} & \multicolumn{1}{c}{assumption of strict additivity (${\bm R} = {\bm R}_1$}) \\ \hline
                     &                                                                      &                                                          \\
${\bm R}_1$          &   $\pm 1.96 \times 2 \frac{s}{\sqrt{N}}$ for all effects             &   $\pm 1.96 \times 2 \frac{s}{\sqrt{N}}$ for all effects \\ \hline
                     &                                                                      &                                                          \\
${\bm R}_2$          &   $\pm 1.96 \times \sqrt{3+\rho} \frac{s}{\sqrt{N}}$ for all effects &   $\pm 1.96 \times 2 \frac{s}{\sqrt{N}}$ for all effects \\ \hline
                     &                                                                      &                                                          \\
${\bm R}_3$          &   $\pm 1.96 \times \sqrt{3-\rho} \frac{s}{\sqrt{N}}$ for $\theta_1$  &   $\pm 1.96 \times 2 \frac{s}{\sqrt{N}}$ for all effects \\
                     &   $\pm 1.96 \times \sqrt{3+\rho} \frac{s}{\sqrt{N}}$ for $\theta_2, \theta_{12}$ &                                                         \\ \hline
                     &                                                                      &                                                          \\
${\bm R}_4$          &   $\pm 1.96 \times \sqrt{3-(\rho_1-\rho_2)} \frac{s}{\sqrt{N}}$ for $\theta_1$,  &                                                         \\
                     &   $\pm 1.96 \times \sqrt{3-(\rho_2-\rho_1)} \frac{s}{\sqrt{N}}$ for $\theta_2$, &  $\pm 1.96 \times 2 \frac{s}{\sqrt{N}}$ for all effects    \\
                     &   $\pm 1.96 \times \sqrt{3-(\rho_1+\rho_2)} \frac{s}{\sqrt{N}}$ for $\theta_{12}$ &                                                     \\ \hline
\end{tabular}
\end{table}

\normalsize

\section{Fisherian randomization-based inference for $2^K$ factorial designs} \label{sec: fisher}

Randomization tests are useful tools because they assess statistical significance of treatment effects from randomized experiments without making any assumption whatsoever about the distribution of the test statistic. Such tests can be used to test Fisher's \emph{sharp null hypothesis} (see Rubin (1980)) of no factorial effect at unit levels, which is a much stronger hypothesis than the traditional one of no average factorial effects. Randomization tests have rarely been studied in the context of factorial experiments, except for the work by Loughin and Noble (1997), who studied such tests in the framework of a linear regression model for the observed response with additive error (in other words, invoking the assumption of strict additivity). In the following paragraphs, based on our proposed framework, we present a mathematical formulation for randomization-based inference procedures for a randomized $2^K$ factorial experiment. In particular, we aim at extending the procedure described by Imbens and Rubin (2012) to obtain ``Fisherian Fiducial'' intervals (Fisher 1930; Wang 2000) for the $J-1$ factorial effects $\tfs$.

To obtain Fisherian intervals, we first obtain estimates $\hattfs$ of the factorial effects from the observed data using (\ref{eq:defpopfactorialest}). The next task is to obtain ranges of plausible values for each factorial effect under the stated randomization scheme, using the above estimates. This is done by (i) considering a sequence of sharp null hypotheses on the factorial effects, (ii) imputing the missing potential outcomes under each sharp null, (iii) computing $p$-values for each factorial effect using their randomization distributions and (iv) identifying ranges of values for the factorial effects for which the $p$ values are not extreme.

Let ${\bm \tau}_i = (\tau_{i1}, \ldots, \tau_{i,J-1})^{\prime}$ denote the $(J-1)$-vector of unit-level factorial effects for unit $i$. Define the sharp null hypothesis $$ H_0^{\bm \eta}: {\bm \tau}_i = {\bm \eta}  \;\  \forall \;\ i = 1, \ldots, N,$$ where ${\bm \eta} = (\eta_1, \ldots, \eta_{J-1})^{\prime}$ is a vector of constants. Let ${\bm G}$ denote the $J \times J$ matrix whose columns are the vectors ${\bm g}_0, \ldots, {\bm g}_{J-1}$. Then from (\ref{eq:defunitfactorial}) it immediately follows that
\begin{equation}
{\bm Y}_i = {\bm G}^{\prime} \left(
                               \begin{array}{c}
                                 \tau_{i0} \\
                                 {\bm \tau}_i \\
                               \end{array}
                             \right),
 \label{eq:linmodelform}
\end{equation}
where $\tau_{i0} = 2^{-K} g_0^{\prime} {\bm Y}_i$ denotes the mean of ${\bm Y}_i$.

Recall that for the $i$th experimental unit, the experimenter observes only one potential outcome $\yo_i$.  Let $\ym_i$ denote the $(J-1)$-vector of the missing potential outcomes. Note that the rows of the $J \times K$ submatrix formed by columns 2 to $K+1$ of ${\bm G}$ represents the $J$ treatment combinations. Denote by ${\bm g}_i^{\mathrm{obs}}$ the row of ${\bm G}$ that contains the treatment combination ${\bm z}$ assigned to unit $i$, and let the submatrix formed by the remaining $J-1$ rows be ${\bm G}_i^{\mathrm{mis}}$. Then from (\ref{eq:linmodelform}) we can write
\begin{equation}
\left(
  \begin{array}{c}
    \yo_i \\
    \ym_i \\
  \end{array}
\right) =
\left(
            \begin{array}{c}
              {\bm g}_i^{\mathrm{obs}} \\
              {\bm G}_i^{\mathrm{mis}} \\
            \end{array}
          \right)
\left(
  \begin{array}{c}
    \tau_{i0} \\
    {\bm \tau}_i/2 \\
  \end{array}
\right) \label{eq:totalpartition} \;.
\end{equation}

Assuming that $H_0^{\bm \eta}$ is true, imputation of the vector of missing potential outcomes $\ym_i$ requires two simple steps:
\begin{enumerate}
\item Estimate $\tau_{i0}$ as $\hat{\tau}_{i0} = Y_i^{\mathrm{obs}} - (1/2) \tilde{\bm g}_i^{\mathrm{obs} \prime} {\bm \eta}$, where $\tilde{\bm g}_i^{\mathrm{obs}}$ is the vector ${\bm g}_i^{\mathrm{obs}}$ without its first element (which is unity).
\item Impute the missing potential outcomes for the $i$th unit using
\[ Y_i^{\mathrm{mis}} = {\bm G}_i^{\mathrm{mis}} \left(
  \begin{array}{c}
    \hat{\tau}_{i0} \\
    {\bm \eta}/2 \\
  \end{array}
\right).
\]
\end{enumerate}

Having generated the complete set of potential outcomes for all $N$ experimental units, we now compute the estimated factorial effects $\hattfs, j=1, \ldots, J-1$, for \emph{each} possible assignment of $N$ experimental units to the $J$ treatment combinations, generating a randomization distribution for each effect. In practice, owing to the prohibitive number of possible arrangements, typically a random sample of the set of all possible assignments is obtained. From the randomization distribution, a $p$-value for each factorial effect is computed by adding the probabilities of all assignments that lead to a value as or more extreme than the value observed for the estimate of that factorial effect.

Let $p(\eta_j)$ denote the $p$-value for the factorial effect $\hattfs$, $j=1, \ldots, J-1$ under the sharp null hypothesis $H_0^{\bm \eta}$. The function $p(\eta_j)$ is monotonic in $\eta_j$. Considering different sharp null hypotheses by varying ${\bm \eta}$, it is possible to obtain $\eta_j^L$ and $\eta_j^U$ satisfying
\begin{eqnarray}
\zeta_j^L &=& \sup_{\eta_j} \{ \eta_j: p(\eta_j) \le \alpha/2 \}, \label{eq:zetaL}\\
\zeta_j^U &=& \inf_{\eta_j} \{ \eta_j: p(\eta_j) \ge 1-\alpha/2 \} \label{eq:zetaU},
\end{eqnarray}
for $0 < \alpha < 1$. Then, $[\zeta_j^L, \zeta_j^U]$ represents $100(1-\alpha)\%$ ``Fisherian'' interval for the factorial effect $\tfs$.

We now illustrate the proposed approach using a simulated $2^2$ experiment with $N = 20$ experimental units. The four potential outcomes for each unit are generated using a multivariate normal model with mean vector $(10,12,13,15)^{\prime}$ and covariance matrix
\[  {\bm R}_2 = \left(
            \begin{array}{cccc}
              1 & .5 & .5 & .5 \\
              .5 & 1 & .5 & .5 \\
              .5 & .5 & 1 & .5 \\
              .5 & .5 & .5 & 1 \\
            \end{array}
          \right).
\]

\begin{table}[htbp]
\small
\caption{Observed outcomes for the 20 experimental units}
\begin{center}
\begin{tabular}{c|r|r|r|r}
\multicolumn{1}{c|}{Unit} & \multicolumn{1}{c|}{$(-1,-1)$} & \multicolumn{1}{c|}{$(-1,1)$} & \multicolumn{1}{c|}{$(1,-1)$} & \multicolumn{1}{c}{$(1,1)$} \\ \hline
 1  &         &         & 11.9445 &         \\ \hline
 2  &         &         & 12.5132 &         \\ \hline
 3  &         &         &         & 15.7970 \\ \hline
 4  &         &         &         & 15.6979 \\ \hline
 5  &         &         &         & 16.1417 \\ \hline
 6  &         &         & 14.7106 &         \\ \hline
 7  &         & 12.7757 &         &         \\ \hline
 8  & 10.1216 &         &         &         \\ \hline
 9  &         & 11.8241 &         &         \\ \hline
10  & 8.9538  &         &         &         \\ \hline
11  & 11.7702 &         &         &         \\ \hline
12  &         & 10.3832 &         &         \\ \hline
13  &         &         & 12.3319 &         \\ \hline
14  &         &         & 13.3999 &         \\ \hline
15  &         &         &         & 13.6675 \\ \hline
16  &         & 13.5522 &         &         \\ \hline
17  &         & 10.1510 &         &         \\ \hline
18  & 10.3421 &         &         &         \\ \hline
19  &         &         &         & 14.0710 \\ \hline
20  & 10.5881 &         &         &         \\ \hline
\end{tabular} \label{tab:obs_outcomes}
\end{center}
\end{table}

The true values of the three estimands are $\bar{\tau}_{.1} = 3$, $\bar{\tau}_{.2} = 2$ and $\bar{\tau}_{.3} = 0$. One set of observed outcomes obtained using a completely randomized treatment assignment is shown in Table \ref{tab:obs_outcomes}. From the observed outcomes, point estimates of the factorial effects are obtained as $\widehat{\bar{\tau}}_{.1} = 2.98$, $\widehat{\bar{\tau}}_{.2} = 1.74$, and $\widehat{\bar{\tau}}_{.3} = 0.36$.

We formulate 100 different sharp null hypotheses of the form ${\bm \tau}_i = {\bm \eta}$ by choosing 100 different vectors ${\bm \eta}$ uniformly from the region ${\mathcal D}_{\eta} = [-6,6]^3$. For each of these sharp null hypotheses, the missing potential outcomes in Table \ref{tab:obs_outcomes} are imputed using the methodology described earlier in this Section. Then, a randomization distribution for each factorial effect is obtained by computing the effect from each of 1000 random draws from the set of possible randomizations. Figure \ref{fig:histograms} shows the randomization distributions of $\bar{\tau}_{.1}$, and $\bar{\tau}_{.2}$ under one of the 100 sharp null hypotheses ${\bm \tau}_i = ( 4.20, -2.22,  0.81)^{\prime}$, with the observed estimated values $\widehat{\bar{\tau}}_{.1}$ and $\widehat{\bar{\tau}}_{.2}$ superimposed as vertical lines. The $p$-values for factorial effects $\bar{\tau}_{.1}$, and $\bar{\tau}_{.2}$ are obtained as 0.891 and 0.000 respectively.

\begin{figure}[ht]
\begin{center}
\psfig{figure=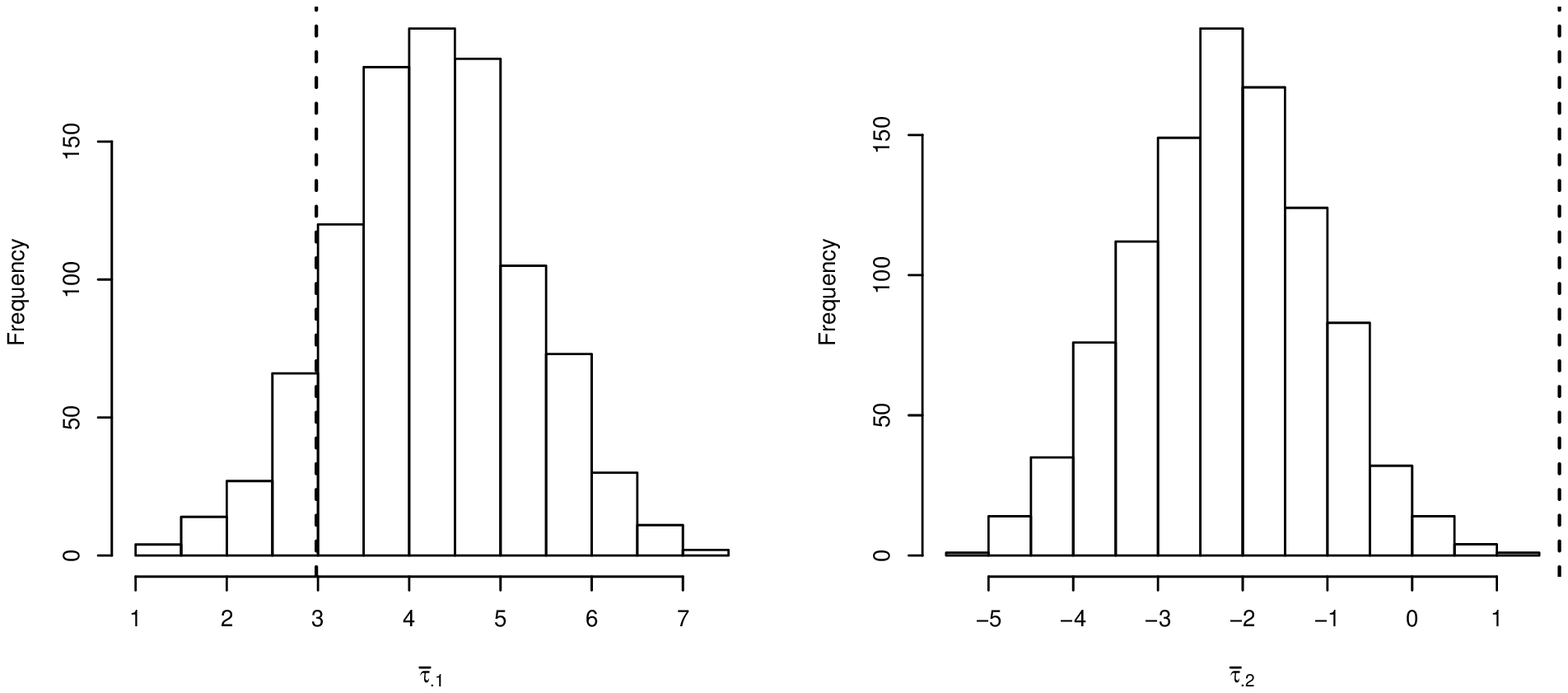,width=5.8in,height=3 in,angle=0}
\end{center}
\caption{Randomization distributions of $\bar{\tau}_{.1}$ and $\bar{\tau}_{.2}$ for ${\bm \eta} = (4.20, -2.22, 0.81)$}  \label{fig:histograms}
\end{figure}

\begin{figure}[htbp]
\begin{center}
\psfig{figure=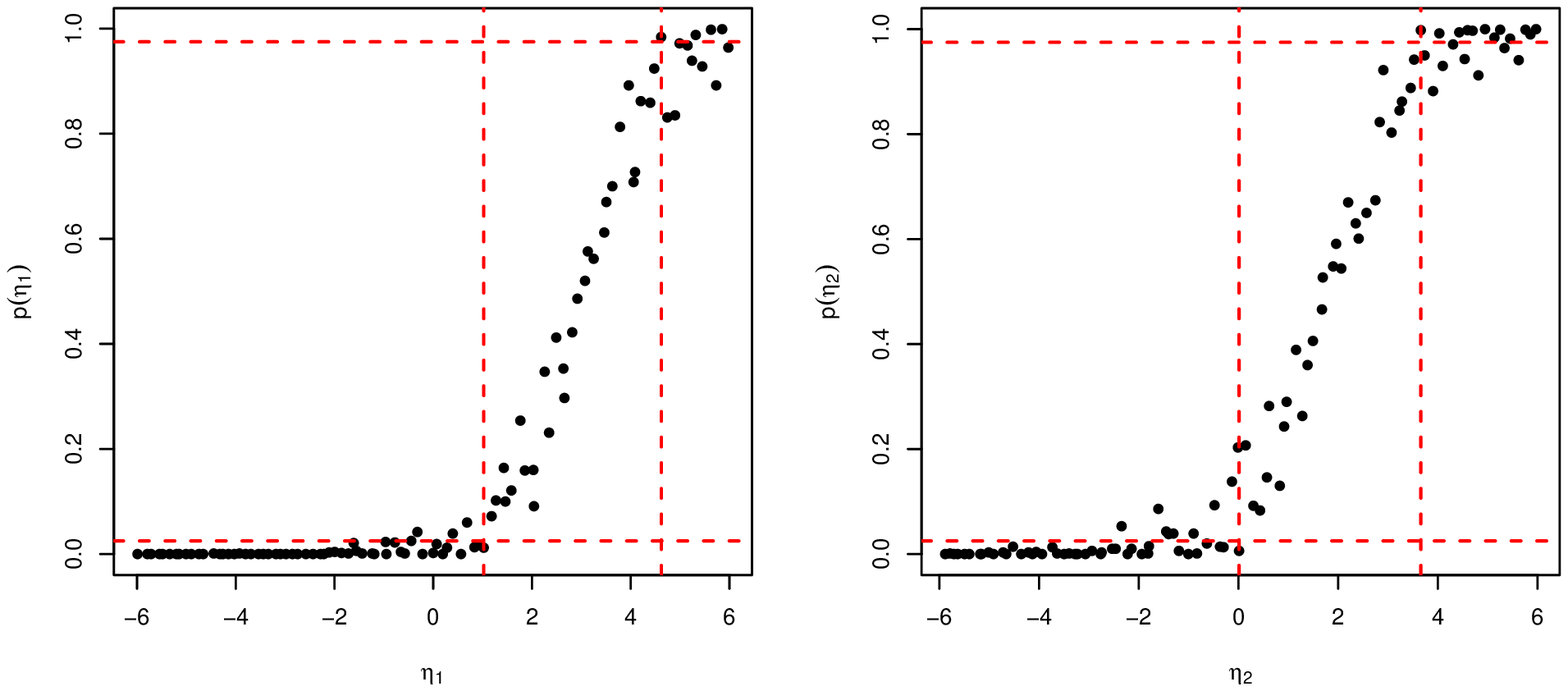,width=5.8in,height=3 in,angle=0}
\end{center}
\caption{Plot of $p(\eta_j)$ versus $\eta_j$ for $j = 1,2$.}  \label{fig:p-val}
\end{figure}

Proceeding in a similar manner, an array of 100 $p$-values are obtained for each of the three factorial effects. The plot of $p(\eta_j)$ versus $\eta_j$ are shown in Figure \ref{fig:p-val} for $j=1$ and $j=2$, i.e., for the two main effects. As described earlier and shown in the plots, the 95\% fiducial limits for the factorial effects are obtained as the maximum and minimum values of $\eta_j$ which are smaller than or equal to 0.025, and larger than or equal to 0.975 respectively.

The Neymanian confidence intervals are computed using (\ref{eq:ci_indiv}). The 95\% Fisherian and Neymanian confidence intervals for the three factorial effects are shown in Table \ref{tab:NeymanFishersummary}.

\small
\begin{table}[htbp]
\begin{center}
\caption{95\% Neymanian and Fisherian ``Fiducial' intervals for the three factorial effects}
\vspace{-0.1 in}
\begin{tabular}{ccccc}
\hline
Factorial effect     &   True value    & Point estimate               &  Neymanian 95\% interval     & Fisherian 95\% interval         \\ \hline
$\bar{\tau}_{.1}$    &    3            &   2.98                       & [1.93,4.03]          &  [1.02,4.61]           \\
$\bar{\tau}_{.2}$    &    2            &   1.74                       & [0.69,2.78]          &  [0.01, 3.65]         \\
$\bar{\tau}_{.3}$    &    0            &   0.36                       & [-0.69,1.40]         &  [-1.38, 1.91]         \\ \hline
\end{tabular}
\end{center} \label{tab:NeymanFishersummary}
\end{table}
\normalsize

We observe that the Fisherian intervals are wider than the Neymanian intervals. However, note that the Fisherian intervals depend on the number of sharp-null hypotheses chosen, and also on the number of randomizations used in each iteration. Comparisons of Fisherian and Neymanian intervals need further investigation and will be reported in future work. Rubin (1984) provided the following Bayesian justification of the Fisherian approach to inference: \emph{it gives the posterior predictive distribution of the estimand of interest under a model of constant treatment effects and fixed units with fixed responses}. Thus, a natural extension of the Fisherian approach is a Bayesian inference procedure described in the following section.

\section{A Bayesian framework} \label{sec:Bayesian}

The proposed Bayesian framework is based on Rubin (1978). The key idea is to develop an imputation model for the missing potential outcomes $\ym$, conditional on the observed outcomes $\yov$ and the observed assignment vector ${\bm W}$. Because the treatment effect is a function of $(\yov,\ym)$, this plan will lead us to obtain the conditional posterior distribution of the factorial effects $\tfs$. To this end, let $f({\bm Y}|\Theta)$ denote a suitable probabilistic model for ${\bm Y}_i$, where $\Theta$ is a vector of parameters endowed with a suitable prior distribution.
For quantities $a,b$ we use the notation $[a|b]$ to denote the probability density
function of $a$ given $b$.
Posterior inference for the factorial effects $\tfs$ entails the following steps:
\begin{enumerate}
\item Obtain $[\Theta|\yov, {\bm W}]$, the  posterior distribution for the parameters $\Theta$ conditional on the observed data $\yov$ and the observed assignment mechanism ${\bm W}$ using Bayes rule.
\item Obtain $[\ym| \yov, \Theta, {\bm W}]$, the conditional posterior distribution of $\ym$ given $\yov$, $\Theta$ and ${\bm W}$. The distribution for $[\ym|\yov,\Theta,{\bm W}]$ is the posterior predictive distribution corresponding to the original model $[{\bm Y}|\Theta]$ (see Rubin(1978)).
\item Next obtain the imputation model $[\ym| \yov, {\bm W}]$  by marginalizing
 the posterior predictive distribution in step 2 over the parameter $\Theta$.
\item Finally, obtain the posterior distribution for the factorial effects $\tfs$.
\end{enumerate}

We now apply the above methodology to balanced factorial designs using a simple hierarchical model with conjugate prior distributions. We will also illustrate some connections between Bayesian inference and Neymanian inference discussed earlier in Section
\ref{sec:neyman}.  We can incorporate the linear model framework introduced in Section
\ref{sec:neyman} into  the following hierarchichal model: % that assumes that each pair of potential outcomes have the same correlation $\rho$:
\be[{modelbayes}]
\left[{\bm Y}_i| {\bm \mu}, \sigma^2\right]  &\stackrel{i.i.d}{=}  \mathrm{N}_{2^K} ({\bm \mu}, \sigma^2 \Omega)\,, \quad \quad
 \Omega =  \left(
            \begin{array}{cccc}
              1 & \rho & \cdots & \rho \\
              \rho & 1 & \cdots & \rho \\
              \cdots & \cdots & \cdots & \cdots \\
              \rho & \rho & \cdots & 1 \\
            \end{array}
          \right).   \\
%{\bm Y}_i  &=  N_{2^k} ({\bm \tau}, \sigma^2 {\bm I}_{2^k}) \;, \\
\left[{\bm \mu}|\sigma^2\right] &= \mathrm{N}_{2^K} ({\bm \mu}_0, \sigma^2/r_0 \, I_{2^K})\;, \\
 \left[\sigma^2\right] &= \mathrm{IG}(\alpha, \beta) \;.
\ee
Thus the interclass correlation for the potential outcomes for each experimental unit is equal to $\rho$, which is considered known.
%The parameter $r_0$ may be intuitively thought of as the `sample size' corresponding to the information contained in the prior.
Here $\mathrm{IG}$ denotes the inverse gamma distribution \footnote{The density of a quantity $\phi$ distributed according to  $\mathrm{IG}(\alpha,\beta)$ is proportional to $\phi^{-\alpha -1} e^{-\phi/ \beta}$.}.
As discussed in Imbens and Rubin (2012),  for moderate sample sizes the analysis will not be sensitive to the choice of priors as long as we are not too dogmatic in eliciting the prior distribution.

The prior distribution  in the hierarchical model \eqref{modelbayes} is chosen to be conjugate and thus posterior
inference and sampling from the posterior predictive distribution is quite straightforward.
Recall $\yob(\bm x)$ from (\ref{eq: ybar}).
For $\bm z \in \mathbb{Z}$, define the quantities
\be
\mm(\bm z) &=  {r\yob(\bm z) + r_0 \bm \mu_0(\bm z) \over
{r + r_0}} \;.
\ee
Also define the vector $\mm_{2^K \times 1} = (\mm(\bm z))_{\bm z \in \mathbb{Z}}$. The posterior distribution of the model parameters $\bm \tau, \sigma^2$ are given by:
 \be [eqn:posttau]
\small[{\bm \mu}| \yov, \sigma^2, \bm W \small] &= \mathrm{N}({\mm}, {\sigma^2 \over r_0 + r} I_{2^K}) \,,\\
 \small[\sigma^2| \yov, \bm W \small] &= \mathrm{IG}\Big(\alpha + 2^{k-1}r \,, \beta + {(r-1) \over 2} \sum_{\bm z} s^2(\bm z) + {1 \over 2} \sum_{\bm z}  {r r_0 \over r + r_0} (\yob(\bm x)-\bm \mu_0(\bm z))^2 %\mm^2(\bm x)
 \Big) \;,
 \ee
 where $s^2(\bm z) = {1 \over (r-1)} \sum_{i: W_i(\bm z) =1 } (\yo_i(\bm z) - \yob(\bm z))^2$.

Note that (\ref{eqn:posttau}) represents step 1 of the four steps involved in making posterior inference of $\tfs$, as described earlier in this Section. After working through steps 2-4 (refer to the proof of Theorem \ref{thm:tmeanvar} in the Appendix), we arrive at the following result on the posterior distribution of the factorial effects:

\vspace{0.1 in}

\begin{thm} \label{thm:tmeanvar}
Assume that the potential outcomes follow the model (\ref{modelbayes}) and fix a factorial effect $\tfs$. Then, for a completely randomized assignment mechanism,
%the posterior distribution of the factorial effect is given by
\be
\E(\tfs | \yo, \bm W) &=  m_j, \label{eqn:tmean}\\
\var(\tfs| \yo, \bm W) &= {V \over 2^{2(K-1)}}\Big({1 \over N} K(\rho)  + {2^{K} \over r + r_0} (1 - {1-\rho \over 2^K})^2
\Big)
%{V \over 2^{2(k-1)}}\Big[ (2^k - 1) (1 + {1 \over r + r_0}) + \rho(2^{k-1} - 1)(2^{k-1}-2) \Big] \;. \label{eqn:tvar}
\ee
where
\be
m_j = \big(1 -{1-\rho \over 2^K} \big) {1 \over 2^{K-1}}\bm g_j' \bm m + {1 - \rho \over 2^K}\,  \hattfs \;,
%{1 \over 2^{K-1}}\Big[\big(1 -{1-\rho \over 2^K} \big)\Big( \sum_{\bm z \in \Zplus_j} \bm m(\bm z) - \sum_{{\bm z} \in \Zminus_j} \bm m(\bm z) \Big) + {1 - \rho \over 2^K}\Big( \sum_{{\bm z} \in \Zplus_j}  \yob(\bm z) - \sum_{{\bm z} \in \Zminus_j} \yob(\bm z) \Big) \Big] \;.
\ee
$V$ denotes the posterior expectation of $\sigma^2$, given by
\be
V = \E(\sigma^2| \yov, \bm W ) = {\beta + {(r-1) \over 2} \sum_{\bm z} s^2(\bm z) + {1 \over 2} \sum_{\bm z}  {r r_0 \over r + r_0} \mm^2(\bm z)  \over  \alpha + 2^{K-1}r  -1} \;,
\ee
and $k(\rho) = \Big((1-\rho^2)(2^K - 1) - 2\rho(1-\rho)(2^{K-1}-1) \Big) $.
 \end{thm}

 Although the expressions for the posterior expectation and variance of $\tfs$ derived in Theorem \ref{thm:tmeanvar} appears somewhat complicated, they reduce to simple forms under specific conditions. This is clear from the following corollary, which also establishes the relationship between the Bayesian and Neymanian estimators of $\tfs$.

\begin{cor} \label{eqn:cormv}
For the model (\ref{modelbayes}),
\be
\lim_{r_0 \rightarrow 0} \mathbb{E}(\tfs|{\bm Y}^{obs}, {\bm W}) &= \hattfs \,,\\
\lim_{\alpha \rightarrow 1, \beta \rightarrow 0}\lim_{r_0 \rightarrow 0} \var(\tfs|{\bm Y}^{obs}, {\bm W})  &=
{4V_\ell \over N} (1 - {1 - \rho \over 2^K}) %= \widehat{\var}(\theta^{\fp}) \,
% \frac{4}{N} s^2, $$ where $s^2 = 2^{-k}\sum_{{\bm x}} s^2({\bm x})$ and $s^2({\bm x})$ is given by (\ref{eq:sampvariance}).
\ee
where $\hattfs$ is the Neymannian estimator for the mean, defined by (\ref{eq:defpopfactorialest}), and $V_\ell = {(r-1) \over r} {1 \over  2^k} {\sum_{\bm x} s^2(\bm x) }$ .
\end{cor}

\begin{remark}
The connection between the Bayesian and Neymanian estimators of $\tfs$, as seen from Corollary \ref{eqn:cormv}, is quite intuitive. Note that $r_0$ can be thought of as a ``prior sample size''. Thus, the condition $r_0 \rightarrow 0$ can be interpreted as lack of prior information on the mean vector of the potential outcomes, or alternatively, as a diffused prior for the mean vector ${\bm \mu}$. The condition $\alpha \rightarrow 1, \beta \rightarrow 0$ leads to a diffuse prior for the variance parameter $\sigma^2$. Collectively, these conditions reflect lack of prior information about the distribution of the potential outcomes.

It is also worthwhile to note that, in the expression for $V_{\ell}$, the multiplicative term ${r-1 \over r} \rightarrow 1$ as the number of replicates $r \rightarrow \infty$. Therefore, $V_{\ell} \rightarrow s^2$ as ${r-1 \over r} \rightarrow 1$, where $s^2 = 2^{-K}\sum_{{\bm x}} s^2({\bm x})$ is an unbiased estimator of the true common variance $S^2$ of all potential outcomes. Thus, as the number of replications $r \rightarrow \infty$, the posterior variance of the factorial effect $\tfs$, corresponding to the reference prior derived in Corollary \ref{eqn:cormv}, approaches the unbiased estimator $\widehat{Var}_{Ney}(\hattfs|{\bm Y})$ defined by (\ref{eq:sethetahat}). Finally, for any finite $r$,  $V_\ell$ is  smaller than its limit $s^2$, which shows that the variance estimate for the finite population estimand is smaller than the variance estimate for the super population estimand.
\end{remark}

\section{An example to demonstrate the advantages of the potential outcomes model in causal inference from $2^K$ factorial experiments}

We now demonstrate the effectiveness and versatility of the proposed framework with an example. Wu and Hamada (2009, Ch. 4) describes a $2^3$ factorial experiment that was originally reported by Box and Bisgaard (1987) and conducted with the objective of reducing the percentage of cracked springs produced by a particular manufacturing process. Three factors -- temperature of the quenching oil, carbon content (percentage) of the steel and the temperature of the steel before quenching -- were investigated. The response for each experimental unit (which represents a spring manufactured under a specific treatment combination ${\bm z}$) is a binary random variable taking the value 1 if the spring does not crack, and zero otherwise.

The data reported from the experiment (see Wu and Hamada 2009, Ch. 4), is actually a summary of the actual experimental data, and reports only the observed percentage of cracked springs for each of the 8 treatment combinations. Here, we assume that $r=100$ springs were manufactured under each treatment combination in a completely randomized sequence (so that $N = 800$). Next, we simulate the potential outcomes for each of the 800 experimental units from a super-population model in which the potential outcome $Y_i({\bm z})$ is drawn as a binary random variable taking value 1 with probability $\pi({\bm z})$, where the probabilities, $\pi({\bm z})$, which are given in Table \ref{tab:springdat}, equal the observed proportion of non-cracked springs, reported in Wu and Hamada (2009, Ch. 4). Finally we generate observed outcomes $Y_i^{\mathrm {obs}}$ for each experimental unit using a completely randomized treatment assignment.

\begin{table}[htbp]
\small
\begin{center}
\begin{tabular}{c | c c c c c c c c|}
${\bm z}$ & (-1,-1,-1) & (-1,-1,1) & (-1,1,-1) & (-1,1,1) & (1,-1,-1) & (1,-1,-1) & (1,1,-1) & (1,1,1) \\ \hline
$\pi({\bm z})$ & 0.67 & 0.79 & 0.61 & 0.75 & 0.59 & 0.90 & 0.52& 0.87 \\
\end{tabular}
\end{center}
\caption{Binary probabilities for the simulation model} \label{tab:springdat}
\end{table}

The traditional method of estimating factorial effects (Wu and Hamada 2009, Ch. 14) from these simulated experimental data is to fit a logistic regression model of the observed response on columns of the $800 \times 8$ design matrix (in which 0's are replaced by -1's, and columns for the interactions are obtained as products of the corresponding factor levels). The factorial effects are estimated as twice the coefficients of the fitted logistic model and are reported in Table \ref{tab:logistic}.

\begin{table}[htbp]
\small
\begin{center}
\begin{tabular}{l | r r r r }
Factorial effect             &  Estimate & Std. Error & $z$-value & $Pr(>|z|)$  \\ \hline
Main effect of 1             &  0.1468   & 0.1753     & 0.84    & 0.4022        \\
Main effect of 2             & -0.2492   & 0.1753     & -1.42   & 0.1551        \\
Main effect of 3             &  1.2975   & 0.1753     &  7.40   & 0.0000        \\
Interaction between 1 and 2  &  0.0787   & 0.1753     & 0.45    & 0.6534        \\
Interaction between 1 and 3  &  0.5426   & 0.1753     & 3.10    & 0.0020       \\
Interaction between 2 and 3  & -0.1060   & 0.1753     & -0.61   & 0.5452         \\
Three-way interaction        & 0.1783    & 0.1753     & 1.02    & 0.3092        \\ \hline
\end{tabular}
\end{center}
\caption{Results from a logistic regression model} \label{tab:logistic}
\end{table}

A natural question is, what factorial effects are the quantities in the second column of Table \ref{tab:logistic} estimating? Let ${\bm \pi}$ denote the vector of $\pi({\bm z})$'s given by the second row of Table \ref{tab:springdat}, and ${\bm \pi}_L$ denote the vector of logistic transforms of $\pi({\bm z})$'s. Then, twice the coefficients of the logistic regression in the second column of Table \ref{tab:logistic} are unbiased estimates of superpopulation-level factorial effects $(1/4){\bm g}_j^{\prime} {\bm \pi}_L $, $j = 1, \ldots, 7$. However, instead of these estimands, the engineers may be interested in superpopulation-level factorial effects $\tss = (1/4){\bm g}_j^{\prime} {\bm \pi}$, i.e., contrasts of $\pi({\bm z})$, rather than their logistic transformations. Point estimates of ${\bm g}_j^{\prime} {\bm \pi}$ can be obtained as transformations of the estimated coefficients of the logistic regression. A more straightforward way of obtaining point estimates of $(1/4){\bm g}_j^{\prime} {\bm \pi}$ is to substitute $p({\bm z}) = \bar{Y}_i^{obs}({\bm z})$ for $\pi({\bm z})$. Asymptotic standard errors of these estimators, given by
$$ \frac{1}{4} \sqrt{\sum_{\bm z} \frac{\pi({\bm z})(1-\pi({\bm z}))}{r}}$$
can again be estimated by substituting $p({\bm z})$ for $\pi({\bm z})$. In the simulated data set, we have the estimated standard error as 0.0304 for each estimated factorial effect, which can be used to construct confidence intervals for the estimands.

Finally, an experimenter may also be interested in finite population-level factorial effects $\tfs = {\bm g}_j^{\prime} {\bm P}$, $j=1, \ldots, 7$, which are contrasts of $P({\bm z}) = \sum_{i=1}^{800} Y_i({\bm z})/800$. Whereas such estimands may not be of much interest in an industrial or a manufacturing scenario, they can be meaningful if we think of a similar experiment in a different setting, e.g., with 800 schools in New York as experimental units, and the three factors as new interventions being investigated by the Department of Education (DOE). Note that the GLM framework and the foregoing asymptotic arguments do not permit us to perform inference for finite population estimands. The Bayesian framework explained in Section \ref{sec:Bayesian}, however, permits us to perform inference for finite population estimands and superpopulation estimands separately.

To illustrate such an analysis, we assume, for the sake of simplicity, that we are interested only in two factorial effects -- the main effect of factor 3, represented by vector ${\bm g}_3$, and the interaction between factors 1 and 3, represented by vector ${\bm g}_5$ -- that appear to be significant from the earlier analysis. We now postulate a simple Hierarchical model. First, assume that all the potential outcomes for the $i$th, $i = 1, \ldots, N$ unit are independently distributed as binary random variables so that their joint distribution is given by
$$f \big({\bm Y}_i|\pi({\bm z}),{{\bm z} \in {\mathbb Z}} \big) = \prod_{{\bm z} \in {\mathcal Z}} \pi({\bm z})^{Y_i({\bm z})} (1-\pi({\bm z}))^{1-Y_i({\bm z})}.$$

\noindent Second, assume that given $\pi({\bm z})$ for all ${{\bm z} \in {\mathbb Z}}$, the potential outcomes of all the $N$ units are independently distributed, so that the joint distribution of the vector of the $Nr$ potential outcomes is given by
$$f \big({\bm Y}|\pi({\bm z}),{{\bm z} \in {\mathbb Z}}\big) = \prod_{1=1}^N f({\bm Y}_i|\pi({\bm z}),{{\bm z} \in {\mathbb Z}}).$$

\noindent Third, assume that $${\bm \pi}_L = \alpha {\bm g}_0 + \beta {\bm g}_3 + \gamma {\bm g}_5,$$
where $(\alpha, \beta, \gamma) \sim  N ({\bm \mu}, {\bm \Sigma})$, and $\mu$ and ${\bm \Sigma}$ are assumed known. For demonstration purposes, we set ${\bm \mu} = (1.048,0.6488,0.2713)$, where the first element 1.048 represents the intercept term of the logistic regression, and the other two are the coefficients of covariate vectors ${\bm g}_3$ and ${\bm g}_5$ in the logistic regression (or half of the estimated factorial effects shown in Table \ref{tab:springdat}). Also, we choose
\[{\bm \Sigma} =  \left(
   \begin{array}{ccc}
     4 & 10 & 10 \\
     10 & 100 & 50 \\
     10 & 50 & 100 \\
   \end{array}
 \right)
\]

\noindent Finally, we note that, because the design is completely randomized, the distribution of ${\bm W}$ is independent of ${\bm Y}$ and ${\bm \pi}$.

To make super population-level inference, we simply need to sample from the posterior distribution
\begin{eqnarray*}
f(\alpha, \beta, \gamma | {\bm Y}^{\mathrm{obs}}, {\bm W}) &\propto& f({\bm Y}^{\mathrm{obs}}| {\bm W}, \alpha, \beta, \gamma) f({\bm W}| \alpha, \beta, \gamma ) f(\alpha, \beta, \gamma) \\ &\propto& f(\alpha, \beta, \gamma) \prod_{{\bm z}} \pi({\bm z})^{\sum_{i:W_i({\bm z})=1}Y_i({\bm z})} (1-\pi({\bm z}))^{r-\sum_{i:W_i({\bm z})=1}Y_i({\bm z})}.
\end{eqnarray*}
The superpopulation level estimands of interest, $\bar{\tau}_{.3}^{\mathrm{SP}}$ and $\bar{\tau}_{.5}^{\mathrm{SP}}$, can be expressed as functions of $\alpha, \beta$ and $\gamma$. After obtaining 10,000 draws from $f(\alpha, \beta, \gamma | {\bm Y}^{\mathrm{obs}}, {\bm W})$ (each draw results in a value of the superpopulation estimands) we obtain point estimates of $\bar{\tau}_{.3}^{\mathrm{SP}}$ and $\bar{\tau}_{.5}^{\mathrm{SP}}$ as their posterior means, and interval estimates as 95\% credible intervals.

To make finite population-level inference, we obtain the imputation model
\begin{eqnarray*}
f({\bm Y}^{\mathrm{mis}} | {\bm Y}^{\mathrm{obs}}, {\bm W}) = \int f({\bm Y}^{\mathrm{mis}} | {\bm Y}^{\mathrm{obs}}, \alpha, \beta, \gamma, {\bm W}) f(\alpha, \beta, \gamma | {\bm Y}^{\mathrm{obs}}, {\bm W}) d\alpha \;\ d\beta \;\ d\gamma,
\end{eqnarray*}
where the finite-population estimands of interest $\bar{\tau}_{.3}$ and $\bar{\tau}_{.5}$ can be expressed as functions of $({\bm Y}^{\mathrm{mis}}, {\bm Y}^{\mathrm{obs}})$. For each draw  of $\alpha, \beta$ and $\gamma$, we obtain ${\bm Y}^{\mathrm{mis}}$, and consequently obtain point and interval estimates of $\bar{\tau}_{.3}$ and $\bar{\tau}_{.5}$. The results are summarized in Table \ref{tab:bayes_results}.

\begin{table}[htbp]
\begin{center}
\begin{tabular}{c|cc||c|cc}
\multicolumn{3}{c||}{Inference for Super Population} & \multicolumn{3}{c}{Inference for Finite Population} \\ \hline
Estimand & \multicolumn{2}{c||}    {Inference using} &  Estimand & \multicolumn{2}{c}    {Inference using} \\
         &   Regression          & Bayesian          &           &     Regression         &   Bayesian      \\ \hline
$\bar{\tau}_{.3}^{\mathrm{SP}}$   &    0.24 [0.16, 0.32]  &  0.24 [0.16, 0.31]  & $\bar{\tau}_{.3}$ & 0.24 [0.16, 0.32] &   0.24 [0.18, 0.30] \\
$\bar{\tau}_{.5}^{\mathrm{SP}}$  &    0.10 [0.02, 0.17]  &  0.10 [0.01, 0.15]  & $\bar{\tau}_{.5}$ & 0.10 [0.02, 0.17] & 0.10 [0.05, 0.13] \\
\end{tabular}
\end{center}
\caption{Summary of inferences of factorial effects 3 and 13 for the super population and finite population} \label{tab:bayes_results}
\end{table}

The analysis described in this Section and the results in Table \ref{tab:bayes_results} reinforce the advantages of the proposed framework over the linear or generalized linear model based inference for factorial effects:

\begin{enumerate}
\item It provides a clear understanding of the estimand.
\item It allows for inference from finite populations, in addition to inference for super population. As seen from Table \ref{tab:bayes_results}, the finite population credible intervals are tighter than the super population credible intervals.
\item It permits making inference for estimands other than the mean of individual potential outcomes. In this example, suppose the experimenter is interested in the median of the unit-level factorial effects $\tau_{ij}$. Because this estimand can be written as a function of $({\bm Y}^{\mathrm{mis}}, {\bm Y}^{\mathrm{obs}})$, making such inference is straightforward.
\end{enumerate}

\section{Concluding Remarks}

In this article we have proposed a framework for causal inference from $2^K$ factorial design using the concept of potential outcomes. Factorial effects are defined for a finite population, and the definitions are extended to estimands for a super population. We have discussed procedures for inferring these estimands under the Neymanian, Fisherian and Bayesian frameworks.  Through these discussions and several examples, we have demonstrated how the utilization of potential outcomes results in better understanding of the estimands and allows greater flexibility in statistical inference of factorial effects, compared to the commonly used linear model based approach.

Other than the benefits already demonstrated in this article, the proposed framework can provide a foundation for addressing complex problems associated with the design and analysis of social, behavioral and medical experiments. Examples of these problems are: (i) unbalanced factorial experiments, (ii) factorial experiments with covariates, (ii) factorial experiments with randomization restrictions, (iii) semi-observational studies with a factorial structure (e.g., experimental units self-select the levels of one or more factors), and (iv) matched sampling in observational studies with a factorial structure. We are therefore hopeful that this article will open up a large number of research problems motivated by complex multi-factor experiments in social, medical and behavioral sciences and help experimenters in making more meaningful causal inference from such experiments.

\begin{center}
{\large\bf Acknowledgments}
\end{center}

This research was supported by the U.S. National Science Foundation grant DMS-1107004, DMS-1107070.
We thank various colleagues and students for useful comments.

\section*{Appendix}
\begin{proof}[Proof of Theorem \ref{thm: mean}]
We argued that $\bar{Y}^{\mathrm{obs}}({\bm z})$ defined by (\ref{eq: ybar}) is an unbiased estimator of $\bar{Y}({\bm z})$. Therefore, the vector $\bar{\bm Y}^{\mathrm{obs}}$ in (\ref{eq:defpopfactorialest}) is an unbiased estimator of $\bar{\bm Y}$ defined by (\ref{eq:ybar(x)}). Comparing (\ref{eq:defpopfactorialest}) and (\ref{eq:altdefpopfactorial}), it immediately follows that $E(\hattfs) = \tfs$.
\end{proof}
\begin{proof}[Proof of Theorem \ref{thm:variance}]

Let $g_{jl}$ denote the $l$th element of vector ${\bm g}_j$, $j = 0, \ldots, J-1$, $l = 1, \ldots, J$, and let $z_l$ denote the $l$th treatment combination, $l = 1, \ldots, J$. From (\ref{eq:defpopfactorialest}), we have
\begin{eqnarray}
\var({\hattfs}) &=& 2^{-2(K-1)} {\bm g}_j^{\prime} \ var(\bar{\bm Y}^{\mathrm{obs}}) \ {\bm g}_j \nonumber \\
&=& 2^{-2(K-1)} \bigg \{\sum_{l=1}^J \ g_{jl}^2 \ var(\bar{Y}^{\mathrm{obs}}(z_l)) + \mathop{\sum \sum}_{l\neq l^{\prime}} \ g_{jl} \ g_{j l^{\prime}} \  cov(\bar{Y}^{\mathrm{obs}}(z_l), \bar{Y}^{\mathrm{obs}}(z_{l^{\prime}}))\bigg \} \nonumber \\
&=& 2^{-2(K-1)} \bigg \{\sum_{l=1}^J \ g_{jl}^2 \frac{N-r}{rN} S^2(z_l) - \frac{1}{N} \mathop{\sum \sum}_{l\neq l^{\prime}} \ g_{jl} \ g_{j l^{\prime}} \  S^2(z_l, z_{l^{\prime}})\bigg \}, \label{eq:thm2_pf1}
\end{eqnarray}
the last step following from (\ref{eq:varybar}) and (\ref{eq:covybar}).
Now, from (\ref{eq:vartheta}), we have
\begin{eqnarray*}
S_j^2 &=& \frac{1}{N-1} 2^{-2(K-1)} \sum_{i=1}^N ({\bm g}_j^{\prime} {\bm Y}_i - {\bm g}_j^{\prime} \bar {\bm Y})^2 \\
&=& \frac{1}{N-1} 2^{-2(K-1)} \sum_{i=1}^N \bigg \{ \sum_{l=1}^J g_{jl} \big( Y_i(z_l) - \bar{Y}(z_l) \big) \bigg \}^2 \\
&=& 2^{-2(K-1)} \bigg \{\sum_{l=1}^J \ g_{jl}^2 S^2(z_l) + \mathop{\sum \sum}_{l\neq l^{\prime}} \ g_{jl} \ g_{j l^{\prime}} \ S^2(z_l, z_{l^{\prime}}) \bigg \}.
\end{eqnarray*}

Consequently,
\begin{equation}
\mathop{\sum \sum}_{l\neq l^{\prime}} \ g_{jl} \ g_{j l^{\prime}} \ S^2(z_l, z_{l^{\prime}}) = 2^{2(K-1)} S_j^2 - \sum_{l=1}^J \ g_{jl}^2 S^2(z_l), \label{eq:thm2_pf2}
\end{equation}

Substituting (\ref{eq:thm2_pf2}) into (\ref{eq:thm2_pf1}), we get
\begin{eqnarray*}
\var({\hattfs}) &=& 2^{-2(K-1)} \bigg \{\sum_{l=1}^J \ g_{jl}^2 \Big( \frac{N-r}{rN} + \frac{1}{N} \Big) \ S^2(z_l) - \frac{1}{N} 2^{2(K-1)} S_j^2 \bigg \} \\
&=& \frac{1}{2^{2(K-1)}r} \bigg \{\sum_{\bm z} S^2 ({\bm z}) - \frac{1}{N} S_j^2  \bigg \},
\end{eqnarray*}
the last step following from the fact that $g_{jl}^2 = 1$ for all $j,l$.
\end{proof}
\begin{proof}[Proof of Theorem \ref{thm:covariance}]
From (\ref{eq:defpopfactorialest}), we have
\begin{eqnarray}
\cov(\hattfs,\hattfss ) &=& 2^{-2(K-1)} {\bm g}_j^{\prime} \ var(\bar{\bm Y}^{\mathrm{obs}}) \ {\bm g}_{j^{\prime}} \nonumber \\
&=& 2^{-2(K-1)} \bigg \{ \sum_{l=1}^J \ g_{jl} g_{j^{\prime}l} \ var(\bar{Y}^{\mathrm{obs}}(z_l))+ \mathop{\sum \sum}_{l\neq l^{\prime}} \ g_{jl} \ g_{j^{\prime} l^{\prime}} \  cov(\bar{Y}^{\mathrm{obs}}(z_l), \bar{Y}^{\mathrm{obs}}(z_{l^{\prime}}))\bigg \} \nonumber \\
&=& 2^{-2(K-1)} \bigg \{\sum_{l=1}^J \ g_{jl} g_{j^{\prime}l} \frac{N-r}{rN} S^2(z_l) - \frac{1}{N} \mathop{\sum \sum}_{l\neq l^{\prime}} \ g_{jl} \ g_{j^{\prime} l^{\prime}} \  S^2(z_l, z_{l^{\prime}})\bigg \}, \label{eq:thm3_pf1}
\end{eqnarray}
the last step following from (\ref{eq:varybar}) and (\ref{eq:covybar}). Now, from (\ref{covartheta}), we have
\begin{eqnarray*}
S_{j j^{\prime}}^2 &=& \frac{1}{N-1} 2^{-2(K-1)} \sum_{i=1}^N ({\bm g}_j^{\prime} {\bm Y}_i - {\bm g}_j^{\prime} \bar {\bm Y})({\bm g}_{j^{\prime}}^{\prime} {\bm Y}_i - {\bm g}_{j^{\prime}}^{\prime} \bar {\bm Y}) \\
&=& \frac{1}{N-1} 2^{-2(K-1)} \sum_{i=1}^N \bigg \{ \sum_{l=1}^J g_{jl} \big( Y_i(z_l) - \bar{Y}(z_l) \big) \bigg \} \bigg \{ \sum_{l=1}^J g_{{j^{\prime}}l} \big( Y_i(z_l) - \bar{Y}(z_l) \big) \bigg \} \\
&=& 2^{-2(K-1)} \bigg \{\sum_{l=1}^J \ g_{jl} g_{j^{\prime} l}  S^2(z_l) + \mathop{\sum \sum}_{l\neq l^{\prime}} \ g_{jl} \ g_{j^{\prime} l^{\prime}} \ S^2(z_l, z_{l^{\prime}}) \bigg \}.
\end{eqnarray*}

Consequently,
\begin{equation}
\mathop{\sum \sum}_{l\neq l^{\prime}} \ g_{jl} \ g_{j^{\prime} l^{\prime}} \ S^2(z_l, z_{l^{\prime}}) = 2^{2(K-1)} S_{j j^{\prime}}^2 - \sum_{l=1}^J \ g_{jl} g_{j^{\prime} l} S^2(z_l), \label{eq:thm3_pf2}
\end{equation}

Substituting (\ref{eq:thm3_pf2}) into (\ref{eq:thm3_pf1}), we get
\begin{eqnarray*}
\var({\hattfs}) &=& 2^{-2(K-1)} \bigg \{\sum_{l=1}^J \ g_{jl} g_{j^{\prime}}  \Big( \frac{N-r}{rN} + \frac{1}{N} \Big) \ S^2(z_l) - \frac{1}{N} 2^{2(K-1)} S_{j j^{\prime}}^2 \bigg \} \\
&=& \frac{1}{2^{2(K-1)}r} \Big [ \sum_{{\bm z} \in \Zminus_j \bigcap \Zminus_{j^{\prime}} } S^2({\bm z}) - \sum_{{\bm z} \in \Zminus_j \bigcap \Zplus_{j^{\prime}} } S^2({\bm z}) - \sum_{{\bm z} \in \Zplus_j \bigcap \Zminus_{j^{\prime}} } S^2({\bm z}) + \sum_{{\bm z} \in \Zplus_j \bigcap \Zplus_{j^{\prime}} } S^2({\bm z}) \Big] - \frac{1}{N}S_{j j^{\prime}}^2
\end{eqnarray*}
and the proof is finished.
\end{proof}

\begin{proof}[Proof of Equation (\ref{eq:s^2thetanon1})] Each of the sets ${\mathbb Z}_j^{-}$ and ${\mathbb Z}_j^{+}$ have cardinality $2^{k-1}$. Thus $2^{k-1} \choose 2$ pairs of distinct treatment combinations can be formed within each of these sets. Because the total number of pairs of distinct combinations is $2^k \choose 2$, the number of pairs in the set ${\mathbb Z}_j^{-} \bigcap {\mathbb Z}_j^{+}$ is ${2^k \choose 2} - 2 {2^{k-1} \choose 2}$. Consequently, from (\ref{eq:s^2theta}), we have
\begin{eqnarray*}
S^2(\hattfs) = \frac{1}{2^{2(k-1)}} \bigg[ {2^k \choose 2} + 2 \rho \bigg \{ 4{2^{k-1} \choose 2} - {2^k \choose 2} \bigg \} \bigg] S^2,
\end{eqnarray*}
and the result follows after simple algebraic manipulations.
\end{proof}

\begin{proof}[Proof of Theorem \ref{thm:tmeanvar}]

Let $\bm z(i)$ denote the treatment combination the $i^{\mathrm{th}}$ observation is assigned to. Then $\ym_i| \yo_i$ is a multivariate Gaussian. For any $\bm z, \bm z^* \neq \bm z(i)$,
\be \label{eqn:fullcond}
\small[\ym_i(\bm z) | \yo_i, \bm \mu, \sigma^2, \bm W \small]  &= \mathrm{N}(\bm \mu(\bm z) + \rho(\yo_i - \bm \mu(\bm z(i))), \sigma^2(1- \rho^2)) \;, \\
\cov(\ym_i(\bm z),\ym_i(\bm z^*)| \yo_i, \bm \mu, \sigma^2,\bm W) &=  \sigma^2 \rho (1- \rho) \;.
\ee
Define the quantities,
\be
\tilde{m}_j  &=
%{1 \over 2^{k-1}} \Big[\big(\sum_{\bm x \in V_1(\theta)} \tau(\bm x)  -  \sum_{\bm x \in V_0(\theta)} \tau(\bm x)  \big)  -
%{(1-\rho) \over 2^{k}} \big(\sum_{\bm x \in V_1(\theta)} (\tau(\bm x) - \yob(\bm x))  -  \sum_{\bm x \in V_0(\theta)} (\tau(\bm x) - \yob(\bm x)) \big) \Big] \;.
{1 \over 2^{K-1}}\Big[\big(1 -{1-\rho \over 2^K} \big) {\bm g}_j' \bm \mu +  {1-\rho \over 2^K}  {\bm g}_j' \yob  \Big] \;, \\
v_j &= {\sigma^2 \over 2^{2(K-1)}} k(\rho) \;.
\ee
Using \eqref{eqn:fullcond} and the definition of $\tau_{ij}$, for every $j$ we see that  $\tau_{ij}$ is a Gaussian random variable
with variance  $v_j$ for all $i$. Furthermore, conditional on $\bm \mu, \sigma^2$ and $\yob$, $\tfs = {1 \over N} \sum_{i=1}^N \tau_{ij}$ has mean $\tilde m_j$.
%\be \label{eqn:thetaidist}
%\small[\theta | \yo_i, \tau, \sigma^2, \bm W \small] = \mathrm{N}\big(\tau_\theta, v_\theta \big) \;.
%\ee
Now we use the crucial fact that, $\tau_{ij}$ and $\tau_{i'j}$ are \emph{independent},
conditional on $\bm \mu, \sigma^2, \yo$ for all $i \neq i'$.
Thus it immediately follows that
\be \label{eqn:thetafin}
 \small[\tfs | \yov, \bm \mu, \sigma^2, \bm W \small] = \mathrm{N}\big(\tilde m_j, {1 \over N} v_j \big) \;.
\ee
Marginalizing over the posterior distribution of $\bm \mu$ and $\sigma^2$, we see that the posterior of a factorial effect $\tfs$ follows a $t$-distribution.  Now the first claim follows from the fact that
$\mathbb{E}(\tfs| \yo, \bm \mu, \sigma^2, \bm W) = \tilde m_j$ (see \eqref{eqn:posttau}) does not depend on $\sigma^2$. Thus by the law of iterated expectations,
\be
\mathbb{E}(\tfs| \yov , \bm W) = \mathbb{E}(\tilde m_j| \yo , \bm W) = m_j.
\ee
Now we turn to the variance of $\tfs$. Since
 $\var( \mathbb{E}(\tfs| \yo,\sigma^2, \bm W)) = 0$,
 it follows that
 $ \var(\tfs| \yov , \bm W) = \mathbb{E}( \var(\tfs| \yov , \sigma^2, \bm W))$.
 Again, using the variance formula for iterated expectations,
\be
\var(\tfs| \yov , \sigma^2, \bm W) &= \mathbb{E}( \var(\tfs| \yov , \bm \mu, \sigma^2, \bm W)) + \var(\mathbb{E}(\tfs| \yov , \bm \mu, \sigma^2, \bm W)) \\
&={1 \over N} {\sigma^2 \over 2^{2(K-1)}} k(\rho)  + {\sigma^2 \over  2^{2(K-1)}}   {2^{K} \over r+r_0} \big(1 -{1-\rho \over 2^K} \big)^2 \label{eqn:sigmapost}
\ee
where for the last term we have used \eqref{eqn:posttau}.
Now the claim follows from \eqref{eqn:sigmapost} and the fact that $\mathbb{E}(\sigma^2| \yov , \bm W])= V$, finishing the proof.
\end{proof}

\begin{proof}[Proof of Corollary \ref{cor:varthetahat_non1}]
 From Theorem \ref{thm:tmeanvar}, $\lim_{r_0 \rightarrow 0} \mathbb{E}(\tfs|\yob, {\bm W})
= \lim_{r_0 \rightarrow 0} m_j $.
\be
 \lim_{r_0 \rightarrow 0} m_j  &= \lim_{r_0 \rightarrow 0} \big(1 -{1-\rho \over 2^K} \big) {1 \over 2^{K-1}}\bm g_j' \bm m + {1 - \rho \over 2^K}\,  \hattfs \\
 &= \big(1 -{1-\rho \over 2^K} \big) \hattfs  + {1 - \rho \over 2^K}\,  \hattfs
 = \hattfs
\ee
where in the penultimate step, we used the fact that $\lim_{r_0 \rightarrow 0} \bm m(\bm z) = \yob(\bm z)$, proving the first claim. For the variance, using Theorem \ref{thm:tmeanvar} and straightforward algebra using the fact
 $\lim_{\alpha \rightarrow 1 , \beta \rightarrow 0}\lim_{r_0 \rightarrow 0} V =  V_\ell$
 yields the result.
\end{proof}

\normalsize

\section*{References}

\small

\begin{description}

\item Box, G. E. P., and Bisgaard, (1987), ``The Scientific Context of Quality Improvement,'' \emph{Quality progress},20, 54-61.

\item Box, G. E. P., and Hunter, J. S.(1961), ``The $2^{k-p}$ fractional factorial designs (Parts 1 and 2),'' \emph{Technometrics}, 3, 311-351, 449-458.

\item Box, G. E. P., Hunter, J. S., and Hunter, G. W. (2005), \emph{Statistics for Experimenters}, New Jersey: John Wiley \& Sons.

\item Cox, D. R. (1958), \emph{Planning of Experiments}, London: Chapman \& Hall.

\item Collins, L., Dziak J., Li, R., (2009) ``Design of Experiments with Multiple Independent Variables: A Resource Management Perspective on Complete and Reduced Factorial Designs,'' \emph{Psychological Methods}, 14.

\item Chakraborty, B., Collins, L., Strecher, V., and Murphy, S. (2009), ``Developing multicomponent interventions using fractional factorial designs,'' \emph{Statistics in Medicine}, 28(21): 2687 - 2708.

\item Fisher, R. A. (1930),``Inverse Probability,'' \emph{Proceedings of the Cambridge Philosophical Society},26,528-535.

\item Fisher, R. A. (1942), \emph{The Design of Experiments}, 3rd ed. New York: Hafner-Publishing.

\item Imbens, G. W. and Rubin, D. B. (2012), \emph{Causal Inference in Statistics and Social Sciences}, in press.

\item Kempthorne, O. (1955), ``The Randomization Theory of Experimental Inference,'' \emph{Journal of The American Statistical Association}, 50, 946–967.

\item Loughin, T.M. and Noble, W. (1997), ``A Permutation Test for Effects in an Unreplicated Factorial Design,'' \emph{Technometrics}, 39, 469-473.

\item McKay, M.D.; Beckman, R.J.; Conover, W.J. (1979), ``A Comparison of Three Methods for Selecting Values of Input Variables in the Analysis of Output from a Computer Code,'' \emph{Technometrics} 21, 239–245.

\item Montgomery, D. C., (2000), \emph{Design and Analysis of Experiments}, New York: John Wiley \& Sons.

\item Morgan, K.L. and Rubin, D. B. (2012), ``Rerandomization to Improve Covariate Balance in Experiments,'' \emph{Annals of Statistics}, 40(2): 1262-1282.

\item Mukerjee, R. and Wu, C. F. J. (2006), \emph{A Modern Theory of Factorial Designs}, New York: Springer.

\item Neyman, J. (1923/1990) ``On the Application of Probability Theory to Agricultural Experiments. Essay on Principles. Section 9,'' \emph{Statistical Science}, 5 (4), 465-472, translated by Dorota M. Dabrowska and Terence P. Speed.

\item Rosenbaum, P., and Rubin, D. B. (1983), ``The Central Role of the Propensity Score in Observational Studies for Causal Effects,'' \emph{Biometrika}, 70, 1, 41-55.

\item Rubin, D. B. (1974), ``Estimating Causal Effects of Treatments in Randomized and Nonrandomized Studies,'' \emph{Journal of Educational Psychology} 66 (5), 688 - 701.

\item Rubin, D. B. (1975), ``Bayesian Inference for Causality: The Importance of Randomization,'' \emph{Proceedings of the Social Statistics Section of the American Statistical Association}, 233-239.

\item Rubin, D. B. (1977), ``Assignment to Treatment Group on the Basis of a Covariate,'' \emph{Journal of Educational Statistics}, 2, 1-26.

\item Rubin, D. B.(1978), ``Bayesian inference for causal effects: The Role of Randomization,'' \emph{The Annals of Statistics}, 6, 34-58.

\item Rubin, D. B. (1980), ``Randomization Analysis of Experimental Data: The Fisher Randomization Test: Comment'', \emph{Journal of the American Statistical Association}, 75, 591-593.

\item Rubin, D. B. (1984), ``Bayesianly Justifiable and Relevant Frequency Calculations for the Applied Statistician,'' \emph{The Annals of Statistics}, 12, 1151-1172.

\item Rubin, D. B. (1990), ``Neyman (1923) and Causal Inference in Experiments and Observational Studies,'' \emph{Statistical Science}, 5, 472-480.

\item Rubin, D. B. (2008), \emph{Statistical Inference for Causal Effects, With Emphasis on Applications in Epidemology and Medical Statistics}, Handbook of Statistics, Volume 27, Elsevier, Oxford, UK.

\item Rubin, D. B. (2010), ``Reflections Stimulated by the Comments of Shadish (2010) and West and Thoemmes (2010),'' \emph{Psychological Methods}, 15, 38-40.

\item Santner, T. J., Williams, B. J. and Notz, W. I. (2003), \emph{The Design and Analysis of Computer Experiments}, New York: Springer.

\item Wang, Y. H. (2000), ``Fiducial Intervals: What Are They?'', \emph{The American Statistician}, 54, 105-111.

\item Wilk, M. B. (1955), ``The Randomization Analysis of a Generalized Randomized Block Design,'' \emph{Biometrika}, 42, 70-79.

\item Wu, C. F. J., and Hamada, M. S. (2009), \emph{Experiments: Planning, Analysis, and Optimization}, New York:
Wiley.

\item Yates, F. (1937), ``The Design and Analysis of Factorial Experiments'', \emph{Imperial Bureau of Soil Sciences - Technical Communication No. 35}, Harpenden.

\end{description}

\end{document}